\documentclass[11pt]{article}

\usepackage{graphicx, verbatim, fullpage, hyperref, amssymb, amsmath, amsthm, enumerate, multicol, algorithm}
\usepackage[margin=1in]{geometry}

\def\final{1}  
\def\iflong{\iffalse}
\ifnum\final=0  
\newcommand{\knote}[1]{[{\tiny karthik: \bf #1}]\marginpar{*}}
\newcommand{\vnote}[1]{[{\tiny Laci: \bf #1}]\marginpar{*}}
\newcommand{\snote}[1]{[{\tiny Santosh: \bf #1}]\marginpar{*}}
\newcommand{\sidecomment}[1]{}
\else 
\newcommand{\knote}[1]{}
\newcommand{\vnote}[1]{}
\newcommand{\snote}[1]{}
\newcommand{\sidecomment}[1]{}
\fi  

\newtheorem{theorem}{Theorem}[section]
\newtheorem{prop}[theorem]{Proposition}
\newtheorem{lemma}[theorem]{Lemma}
\newtheorem{claim}[theorem]{Claim}
\newtheorem{corollary}[theorem]{Corollary}

\def\F{\mathcal{F}}

\def\L{\mathcal{L}}
\def\K{\mathcal{K}}
\def\H{\mathcal{H}}

\def\C{\mathcal{C}}
\def\O{\mathcal{O}}
\def\R{\mathbb{R}}

\def\Z{\mathbb{Z}}
\def\V{\mathcal{V}}

\def\eps{\epsilon}
\def\supp{\text{supp}}
\def\o{\text{odd}}

\newcommand{\fc}[1]{#1\text{-critical}}
\newcommand{\pfc}[1]{#1\text{-positively-critical}}

\newcommand{\proofbox}{\hfill$\square$}

\title{The Cutting Plane Method is Polynomial for Perfect
  Matchings \thanks{This work was supported in part by NSF award
    AF0915903, and the second author was also supported by NSF Grant
    CCF-0914732. This work was done while the first two authors were
    affiliated with  the College of Computing, Georgia Institute of Technology. }}

\author{Karthekeyan Chandrasekaran \footnote{School of Engineering and
    Applied Sciences, Harvard University. Email:
    karthe@seas.harvard.edu.
}
\and
L\'aszl\'o A. V\'egh \footnote{Department of Management, London School
  of Economics. Email: l.vegh@lse.ac.uk.
}
\and
Santosh S. Vempala \footnote{College of Computing, Georgia Institute of Technology. Email: vempala@cc.gatech.edu.}
}

\date{}

\begin{document}
\maketitle

\begin{abstract}
The cutting plane approach to finding minimum-cost perfect matchings has been discussed by
several authors
over past decades \cite{Padberg82,Grotschel85,Lovasz86,Trick87,fischetti07}, and
its convergence has been an open question.
We give a cutting plane algorithm that converges in polynomial-time using only Edmonds' blossom inequalities; it maintains 
half-integral intermediate LP solutions supported by a disjoint union of odd cycles and edges.  
Our main insight is a method to retain only a subset of the previously added cutting planes based on their dual values. 
This allows us to quickly find violated blossom inequalities and argue convergence by tracking the number of odd cycles in the support of intermediate solutions.
\end{abstract}

\section{Introduction}
Integer programming is a powerful and widely used approach for modeling and solving
discrete optimization problems \cite{nemhauser-wolsey,schrijver-IP-book}. Not
surprisingly, it is NP-complete and the fastest known algorithms are exponential in the
number of variables (roughly $n^{O(n)}$ \cite{Kannan87}). In spite of this
intractability, integer programs of considerable sizes are routinely solved in
practice. A popular approach is the cutting plane method, proposed by Dantzig,
Fulkerson and Johnson \cite{CP-intro-TSP-DFJ54} and pioneered by Gomory \cite{Gomory58,
gomory-cutting-plane-G60, Gomory63}. This approach can be summarized as follows:
\begin{enumerate}
\item Solve a linear programming relaxation (LP) of the given integer program (IP) to
    obtain a basic optimal solution $x$.
\item If $x$ is integral, terminate. If $x$ is not integral, find a linear inequality
    that is valid for the convex hull of all integer solutions but violated by $x$.
\item Add the inequality to the current LP, possibly drop some other inequalities and
    solve the resulting LP to obtain a basic optimal solution $x$. Go back to Step 2.
\end{enumerate}

For the method to be efficient, we require the following: {(a)}  an efficient procedure
for finding a violated inequality (called a cutting plane), {(b)} convergence of the
method to an integral solution using the efficient cut-generation procedure and {(c)} a
bound on the number of iterations to convergence. Gomory gave the first efficient
cut-generation procedure and showed that the cutting plane method implemented using
his procedure always converges to an integral solution \cite{Gomory63}. Today, there
is a rich theory on the choice of cutting planes, both in general and for specific
problems of interest.
This theory includes interesting families of cutting planes with efficient
cut-generation procedures
\cite{Gomory58,intersection-cuts-B71,cg-rank-C73,disjunctive-cuts-B79,split-cuts-CKS90,
MIR-cuts-NW90, lift-and-project-BCC93, lovasz-schrijver91,sherali-adams94}, valid
inequalities, closure properties and a classification of the strength of inequalities
based on their {\em rank} with respect to cut-generating procedures
\cite{cutting-plane-closures-CL00} (e.g., the Chv\'atal-Gomory rank
\cite{cg-rank-C73}), and testifies to the power and generality of the cutting plane
method.

To our knowledge, however, there are no polynomial bounds on the number of iterations
to convergence of the cutting plane method even for specific problems using specific
cut-generation procedures. The best bound for general $0$-$1$ integer programs remains
Gomory's bound of $2^n$
\cite{Gomory63}. It is possible that such a bound can be significantly
improved for IPs with small Chv\'atal-Gomory rank
\cite{fischetti07}. A more realistic possibility is that the approach
is provably efficient for combinatorial optimization problems that are
known to be solvable in polynomial time.
An ideal candidate could be a problem that (i) has a
polynomial-size IP-description (the LP-relaxation is polynomial-size), and (ii)
the convex-hull of integer solutions has a polynomial-time separation oracle.
Such a problem admits a polynomial-time algorithm via the Ellipsoid method \cite{GLS}.
Perhaps the first such interesting problem is  minimum-cost perfect matching:
{\em given a graph with costs on the edges, find a perfect matching of
  minimum total cost.} This is a very well-studied problem with efficient algorithms \cite{Lovasz86, Schrijver03}.

A polyhedral characterization of the matching problem was discovered by Edmonds
\cite{Edmonds65}.  Basic solutions (extreme points of
the polytope) of the following linear program correspond to perfect matchings of the graph.
\begin{equation}\tag{P}\label{prog:P-PM}
\begin{aligned}
\min& \sum_{uv\in E} c(uv) x(uv) \\
x(\delta(u))&=1\quad\forall u\in V\notag\\
x(\delta(S))&\ge 1\quad \forall S \subsetneq V, |S| \mbox{ odd}, 3\le|S|\le|V|-3
\notag\\
x&\ge0\notag
\end{aligned}
\end{equation}
The relaxation with only the degree and nonnegativity constraints,
known as the {\em bipartite relaxation}, suffices to characterize the convex-hull of
perfect matchings in bipartite graphs, and serves as a natural starting relaxation. The
inequalities corresponding to sets of odd cardinality greater than $1$ are called {\em
blossom} inequalities. These inequalities have Chv\'atal rank 1, i.e., applying one
round of all possible Gomory cuts to the bipartite relaxation suffices to recover the
perfect matching polytope of any graph \cite{cg-rank-C73}. Moreover, although the
number of blossom inequalities is exponential in the size of the graph, for any point
not in the perfect matching polytope, a violated (blossom) inequality can be found in
polynomial time \cite{Padberg82}. This suggests a natural cutting plane algorithm
(Algorithm \ref{alg:cutting-plane}), proposed by Padberg and Rao \cite{Padberg82} and
discussed by Lov\'asz and Plummer \cite{Lovasz86}. Experimental evidence suggesting that this method converges quickly
was given by Gr\"{o}tschel and Holland \cite{Grotschel85},  by Trick \cite{Trick87},
and by Fischetti and Lodi \cite{fischetti07}. It has been open to rigorously explain
their findings. In this paper, we address the question of whether the method can be
implemented to converge in polynomial time.

\begin{algorithm}[ht]
\caption{Cutting plane method for perfect matching}
\label{alg:cutting-plane}
\begin{enumerate}
\item Start by solving the bipartite relaxation.
\item While the current solution has a fractional coordinate,
\begin{enumerate}
	\item Find a violated blossom inequality and add it to the LP.
	\item Solve the new LP.
\end{enumerate}
\end{enumerate}
\end{algorithm}

Known polynomial-time algorithms for minimum-cost perfect matching are
variants of Edmonds' weighted matching algorithm
\cite{Edmonds65}. A natural idea is to interpret the Edmonds' algorithm, which maintains a partial matchng and shrinks and unshrinks odd sets, as a cutting plane algorithm, possibly
by adding cuts corresponding to the shrunk sets in the iterations of Edmonds'
algorithm. However, there seems to be no correspondence between LP solutions and partial matchings and shrunk sets in his algorithm. 
It is even possible that the initial
bipartite relaxation already has an integer optimal solution, whereas Edmonds'
algorithm proceeds by shrinking and deshrinking a long sequence of odd sets.
So we take a different route.

The bipartite relaxation has the nice property that any
basic solution is half-integral and its support is a disjoint union of edges and odd
cycles. This makes it particularly easy to find violated blossom inequalities---any
odd component of the support gives one. This is also the simplest heuristic that is
employed in the implementations \cite{Grotschel85,Trick87} for finding violated blossom
inequalities. However, if we have a fractional solution in a later phase, there is no
guarantee that we can find an odd connected component whose blossom inequality is
violated, and therefore sophisticated and significantly slower separation methods are
needed for finding cutting planes, e.g., the Padberg-Rao procedure \cite{Padberg82}.
Thus, it is natural to wonder if there is a choice of cutting planes that maintains
half-integrality of intermediate LP optimal solutions.

\begin{figure}[ht]
\centering
\begin{tabular}{ccccccc}
\includegraphics[scale=0.4]{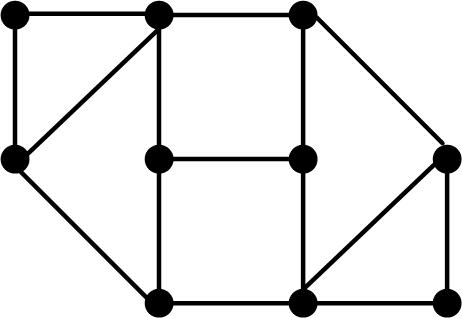}
& & &
\includegraphics[scale=0.4]{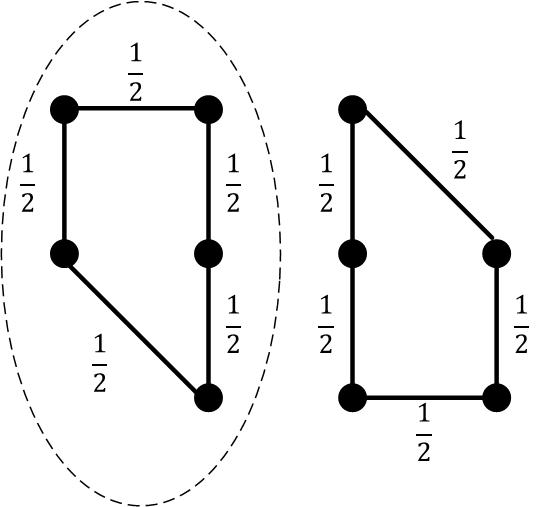}
& & &
\includegraphics[scale=0.4]{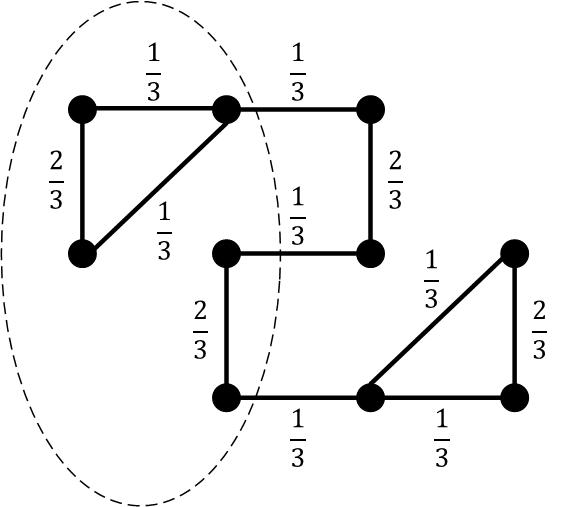}\\
Graph $G$ with all & & & The starting optimum $x_0$& & & Basic feasible solution obtained \\
edge costs one& & &  and the cut to be imposed & & & after imposing the cut
\end{tabular}
\caption{An example where half-integrality is not preserved by the black-box LP solver.}\label{fig:five-cycles}
\end{figure}

At first sight, maintaining half-integrality using a black-box LP solver seems to be impossible.
Figure~\ref{fig:five-cycles} shows an example where the starting solution consists of
two odd cycles. There is only one reasonable way to impose cuts, and it leads to a basic feasible solution that is not half-integral. We observe however, that in the example, the bipartite relaxation also has an integer optimal solution. The issue here seems to be the existence of multiple basic optimal solutions. To avoid such degeneracy, we will ensure that all linear systems that we encounter have unique optimal solutions.

This uniqueness is achieved by a simple deterministic perturbation of the integer
cost function, which increases the input size polynomially. However, this
perturbation is only a first step towards maintaining half-integrality of intermediate
LP optima. As we will see presently, more careful cut retention and cut addition procedures are needed even to
maintain half-integrality.

\subsection{Main result}
To state our main result, we first recall the definition of a laminar family:
A family $\F$ of subsets of $V$ is called {\em laminar}, if for any two sets 
$X,Y\in {\cal F}$, either $X\cap Y=\emptyset$ or $X\subseteq Y$ or $Y\subseteq X$.

Next we define a perturbation to the cost function that will help avoid some degeneracies.
Given an integer cost function $c:E\rightarrow \Z$ on
the edges of a graph $G=(V,E)$, let us define the perturbation $\tilde
c$ by ordering the edges arbitrarily, and increasing the cost of edge
$i$ by $1/2^{i}$.

We are now ready to state our main theorem.

\begin{theorem}\label{thm:main}
Let $G=(V,E)$ be a graph on $n$ nodes with edge costs $c:E \rightarrow \Z$ and let
$\tilde{c}$ denote the perturbation of $c$. Then,
there exists an implementation of the cutting plane method that finds the minimum
$\tilde{c}$-cost perfect matching such that
\begin{enumerate}[(i)]
\item every intermediate LP is defined by the bipartite relaxation constraints and a
    collection of blossom inequalities corresponding to a laminar family of odd
    subsets,
\item every intermediate LP optimum is unique, half-integral, and supported by a
    disjoint union of edges and odd cycles and
\item the total number of iterations to arrive at a minimum $\tilde{c}$-cost perfect
    matching is $O(n \log n)$.
\end{enumerate}
The collection of blossom inequalities used at each step can be identified by
solving an LP of the same size as the current LP. Further, the minimum $\tilde{c}$-cost perfect
matching is also a minimum $c$-cost perfect matching.
\end{theorem}

To our knowledge, this is the first polynomial bound on the convergence of a cutting plane method for matchings using a black-box LP solver.
It is easy to verify that for an $n$-vertex graph, a laminar family of nontrivial odd sets may
have at most $n/2$ members, hence every intermediate LP has at most $3n/2$
inequalities apart from the non-negativity constraints. This ensures that the intermediate LPs do not blow-up in size. While the LPs could be solved using a black-box LP solver, we also provide a combinatorial algorithm that could be used to solve them.

\subsection{Related work}

\snote{To be edited!}
We discovered after completing this work that a very similar question
was addressed by Bunch in his thesis \cite{Bunch97}.
He proposed a cutting plane algorithm for the more general
$b$-matching problem, maintaining a sequence of half-integral primal
solutions.
The intermediate LPs are solved using a primal-dual simplex method. Each new cut and simplex pivot is chosen carefully so that the primal/dual solution  resulting after the pivot step can also be obtained through a combinatorial operation in an associated graph. So, the intermediate solutions correspond to the intermediate solutions of a combinatorial algorithm. This combinatorial algorithm is a variant of 
Miller-Pekny's combinatorial algorithm \cite{Miller95} that proceeds by maintaining a sequence of half-integral solutions and
 is known to terminate in polynomial-time.
Consequently, the intermediate primal solutions of Bunch's algorithm are also half-integral and the algorithm terminates in polynomial time.

The main advantages of
our algorithm compared to Bunch's work are 
the following:
\begin{itemize}
\item We present a purely cutting plane method, using a black-box LP
  solver.  In contrast, the cut generation method in Bunch's algorithm,
  similar to Gomory cuts, relies heavily on the optimal simplex
  tableaux that is derived by the primal-dual simplex method; this
  mimics a certain combinatorial algorithm.
  We will also refer to a similar combinatorial algorithm, however, it
  will only be used in the analysis.
\item We provide a simple and concise sufficient condition  for the LP
  defined by the bipartite relaxation constraints and a subset of
  blossom inequalities to have a half-integral optimum (Lemma
  \ref{lem:pof-uniqueness-implies-half-integrality}). No such insight
  is given in \cite{Bunch97}.
\item Our treatment is substantially simpler, both regarding the algorithm and the proofs.
We prove a bound $O(n\log n)$ on the number of cutting plane iterations, whereas no explicit bound is given in \cite{Bunch97}.
\end{itemize}


\subsection{Cut selection via dual values}
Uniqueness of optimum LP
solutions does not suffice to maintain half-integrality of optimal solutions upon
adding any sequence of blossom inequalities. At any iteration, inequalities that are tight for the current optimal solution are
natural candidates to be retained in the next iteration while the new inequalities are
determined by odd cycles in the support of the current optimal solution. However, it
turns out that keeping all tight inequalities does not maintain half-integrality. 
In fact, as shown in Fig. \ref{fig:five-cycles}, even a laminar family of blossom
inequalities is insufficient to guarantee the nice structural property on the
intermediate LP optimum. Thus the new cuts have to be added carefully 
and it is also crucial that we choose carefully which older cuts to retain.
\snote{We could add the second non-half-integrality example here that holds even with uniqueness.}\knote{The example is too large. I have a picture of it. It has more than 20 nodes. I don't think it is insightful for the reader.}

Our main algorithmic insight is that the choice of cuts for the next iteration can be
determined by examining optimal dual solutions to the current LP --- we retain those cuts whose dual values are strictly positive. Since there could be multiple dual optimal solutions, we use a restricted type of dual optimal solution (called {\em positively-critical dual} in this paper) that can be computed either by solving a single LP of the same size or
via a combinatorial subroutine. We ensure that the set of cuts imposed in any LP are
laminar and correspond to blossom inequalities.
We remark that eliminating cutting planes that have zero dual values 
is common
in implementations of the cutting plane algorithm.  

\subsection{Algorithm C-P-Matching}
All graphs in the paper will be undirected. For a graph $G=(V,E)$, and a subset $S\subseteq V$, let $\delta(S)$ denote the set of edges in $E$ with exactly one endpoint in $S$, and $E[S]$ the set of edges in $E$ with both endpoints inside $E$. For a node $u\in V$, $\delta(u)$ will be used to denote $\delta(\{u\})$, the set of edges incident to $u$. For a vector $x:E\rightarrow \R_+$, $\supp(x)$ will denote its support, i.e., the set of edges $e\in E$ with $x(e)>0$.

We now describe our cutting plane algorithm. Let $G=(V,E)$ be a graph, $c:E\rightarrow\R$ a cost function on the
edges, and assume $G$ has a perfect matching. The {\em
  bipartite relaxation} polytope and its dual are specified as
follows.

\begin{multicols}{2}
\noindent
\begin{equation}\tag{$P_0(G,c)$}\label{prog:P-bip}%
\begin{aligned}%
\min& \sum_{uv\in E} c(uv) x(uv) \\
x(\delta(u))&=1\quad\forall u\in V\\
x&\ge0
\end{aligned}
\end{equation}
\begin{equation}\tag{$D_0(G,c)$}\label{prog:D-bip}%
\begin{aligned}%
\max\sum_{u\in V}\pi(u)& \\
\pi(u)+\pi(v) &\le c(uv) \quad \forall uv\in E\\
~
\end{aligned}
\end{equation}
\end{multicols}
We call a vector $x \in \R^E$ {\em proper-half-integral} if $x(e) \in \{0,1/2,1\}$ for
every $e \in E$, and its support $\supp(x)$ is a disjoint union of edges and odd cycles.
The bipartite relaxation of any graph has the following well-known property.

\begin{prop}\label{prop:proper}
Every basic feasible solution $x$ of \ref{prog:P-bip} is proper-half-integral.
\proofbox
\end{prop}

Let $\O$ be the set of all odd subsets of $V$ of size at least $3$, and let $\V$ denote the set of one element
subsets of $V$.
For a family of odd sets ${\cal F}\subseteq {\cal O}$, consider the
following pair of linear programs.

\begin{multicols}{2}
\noindent
\begin{align}
\min& \sum_{uv\in E} c(uv) x(uv) \tag{$P_{\F}(G,c)$}\label{prog:P-F}\\
x(\delta(u))&=1\quad\forall u\in \V\notag\\
x(\delta(S))&\ge 1\quad \forall S\in {\cal F}\notag\\
x&\ge0\notag
\end{align}
\begin{align}
\max 
\sum_{S\in\V\cup {\cal F}}\Pi(S)& \tag{$D_{\cal F}(G,c)$}\label{prog:D-F}\\
\sum_{S\in\V\cup {\cal F}:uv\in \delta(S)} \Pi(S)&\le c(uv) \quad \forall
uv\in E\notag \\
\Pi(S)&\ge0\quad\forall S\in \F\notag
\end{align}
\end{multicols}
Note that $P_\emptyset(G,c)$ is identical to
\ref{prog:P-bip}, whereas $P_{\cal O}(G,c)$  is identical
to (\ref{prog:P-PM}). Every intermediate
LP in our cutting plane algorithm will be \ref{prog:P-F} for some
laminar family $\cal F$.
We will use $\Pi(v)$ to denote $\Pi(\{v\})$ for dual solutions.

\medskip
Assume we are given a dual feasible solution $\Gamma$ to $D_{\F}(G,c)$. We say that a
dual optimal solution $\Pi$ to $D_{\F}(G,c)$ is {\em $\Gamma$-extremal}, if it
minimizes
\[h(\Pi,\Gamma)=\sum_{S\in \V\cup\F}\frac{|\Pi(S)-\Gamma(S)|}{|S|}\]
among all dual optimal solutions $\Pi$. A $\Gamma$-extremal dual optimal solution can
be found by solving a single LP if we are provided with the primal optimal solution to
\ref{prog:P-F}  (see Section \ref{sec:extremal-dual-solutions}).

Our proposed  cutting plane implementation is Algorithm \ref{alg:main-alg}.
From the current set of cuts, we retain only those which have a positive value in an
extremal dual optimal solution; let $\H'$ denote this set of cuts.
The new set of cuts $\H''$ correspond to odd cycles in
the support of the current solution. However, in order to maintain
laminarity of the cut family, we do not add the vertex sets of these
cycles but instead  their union with all the sets in $\H'$ that they intersect. We
will show that these unions are also odd sets and thus give blossom inequalities. It will follow from our analysis
that each set in $\H'$ intersects at most new one odd cycle, so the sets added in $\H''$ are disjoint. 
This step is illustrated in Fig. \ref{fig:algo}.
(In the first iteration, there is no need to solve the dual LP as $\F$ will be
empty.)

\begin{algorithm}[ht]\caption{Algorithm C-P-Matching} \label{alg:main-alg}
\begin{enumerate}
\item Let $c$ be the cost function on the edges after perturbation (i.e., after
    ordering the edges arbitrarily and increasing the cost of edge $i$ by $1/2^i$).
\item {\bf Initialization.}  $\F=\emptyset$, $\Gamma\equiv 0$.
\item {\bf Repeat} until $x$ is integral:
\begin{enumerate}
\item {\bf Solve LP.} Find an optimal solution $x$ to $P_{\F}(G,c)$.
\item {\bf Choose old cutting planes.} Find a $\Gamma$-extremal dual optimal solution
    $\Pi$ to $D_{\F}(G,c)$.
 Let
\[\H'=\{S\in {\F}: \Pi(S)>0\}.\]
\item {\bf Choose new cutting planes.} Let $\C$ denote the set of odd cycles in
    $\supp(x)$. For each $C\in \C$, define $\hat C$ as the union of $V(C)$ and the
    maximal sets of $\H'$ intersecting it. Let
\[\H''= \{\hat C: C\in \C\}.\]
\item Set the next $\F=\H'\cup \H''$ and $\Gamma=\Pi$.
\end{enumerate}
\item {\bf Return} the minimum-cost perfect matching $x$.\\
\end{enumerate}
\end{algorithm}

\begin{figure}[ht]
\centering
\includegraphics[scale=0.45]{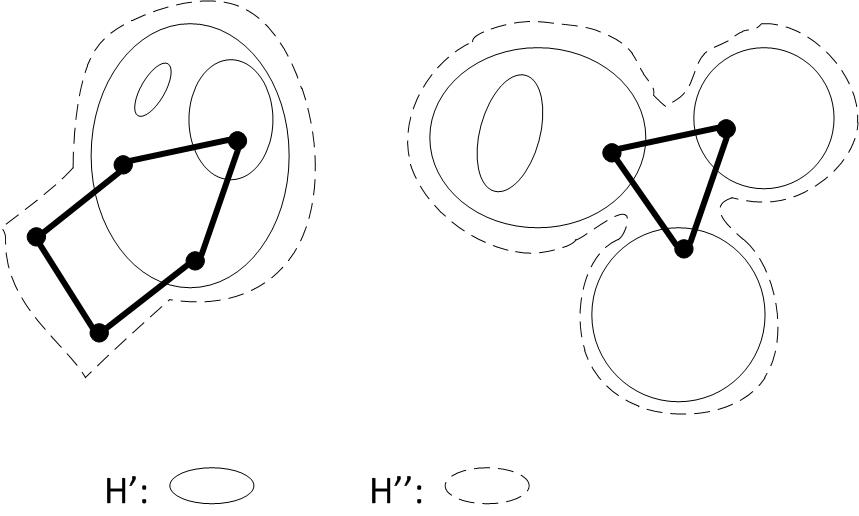}
\caption{Adding laminar sets for new cycles.}\label{fig:algo}
\end{figure}

\subsection{Overview of the analysis}
The aim of our analysis is two-fold: to show that the algorithm maintains half-integrality and that it converges quickly.

Our analysis to show half-integrality is based on extending the notion of factor-criticality.  Recall that a graph is {\em factor-critical} if deleting any node leaves the graph with a perfect matching. Factor-critical graphs play an important role in matching algorithms
(for more background, we refer to the books \cite{Lovasz86} and \cite{Schrijver03}). 
As an important example, the sets contracted during the course of Edmonds' matching algorithms (both unweighted \cite{Edmonds65matching} and weighted \cite{Edmonds65}) are factor-critical subgraphs. 
We define a notion of factor-criticality for weighted graphs that also takes a laminar odd family $\F$ into account.

To prove convergence, we use the number
of odd cycles in the support of an optimal half-integral solution as a
potential function. We first show
$\o(x_{i+1})\le \o(x_i)$, where $x_i,x_{i+1}$ are consecutive LP optimal solutions, and
$\o(x)$ is the number of odd cycles in the support of $x$. We further show that the cuts added
in iterations where $\o(x_i)$ does not decrease continue to be retained until $\o(x_i)$
decreases. Since the maximum size of a laminar family of nontrivial odd sets is $n/2$,
we get a bound of $O(n\log{n})$ on the number of iterations.

The analysis of the potential function behavior is quite intricate. It proceeds by
designing a {\em half-integral} combinatorial procedure for
minimum-cost perfect matching, and arguing that the optimal solution to the extremal
dual LP must correspond to the one found by this procedure. 
We emphasize that this procedure is used only in the analysis.
(It could also be used as a combinatorial method to solve the intermediate LPs in place of a black-box LP solver.) 
A complete, stand-alone extension of the half-integral combinatorial procedure to obtain min-cost perfect matchings is given in \cite{our-matching-alg}.
\section{Factor-critical sets}
In this section, we define factor-critical sets and factor-critical
duals. 

Let $H=(V,E)$ be a graph and $\F\subseteq\O$ be a laminar family of odd subsets of
$V$. We say that an edge set $M\subseteq E$ is an {\em $\F$-matching},
if it is a matching, and for any $S\in\F$, $|M\cap \delta(S)|\le
1$. For a set $U\subseteq V$, we call a set $M$ of edges to be an {\em
  $(U,\F)$-perfect-matching}, if it is an $\F$-matching covering
precisely the vertex set $U$.

A set $S\in \F$ is defined to be {\em $(H,\F)$-factor-critical} or {\em $\F$-factor-critical} in the graph $H$, if for every node
$u\in S$, there exists an $(S\setminus\{u\},\F)$-perfect-matching using the edges of
$H$. For a laminar family $\F$ and a feasible solution $\Pi$  to $D_{\F}(G,c)$, let
$G_\Pi=(V,E_\Pi)$ denote the graph of tight edges. For simplicity we will say that
a set $S \in \F$ is $(\Pi,\F)$-factor-critical if it is $(G_\Pi,\F)$-factor critical, i.e.,
$S$ is $\F$-factor-critical in $G_\Pi$. For a vertex $u \in S$,
the corresponding matching $M_u$ is called the {\em $\Pi$-critical-matching} for $u$. (If there are multiple such matchings, select $M_u$ arbitrarily.)
If $\F$ is clear from the context, then we simply say $S$ is
$\Pi$-factor-critical. 

Fig. \ref{fig:factor-critical} gives an example of an $(H,\F)$ factor-critical set. The three sets in $\F$ besides $S$ are indicated with circles. For any vertex $u \in S$, deleting $u$ leaves an $(S\setminus u, \F)$-perfect matching. However, when the edge $e$ is removed, the graph on the vertex set $S$ remains factor-critical, but the only perfect matching for $S\setminus u$ has two edge crossing $T$, a set in $\F$, and therefore $S$ is not $(H,\F)$-factor-critical.
\begin{figure}[ht]
\centering
\begin{tabular}{cc}
\includegraphics[scale=0.4]{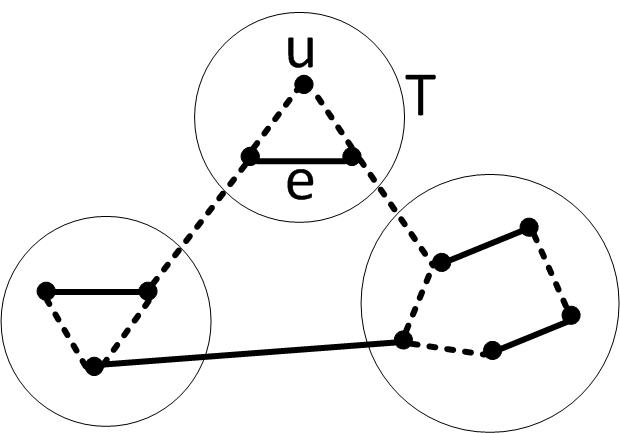}
& 
\includegraphics[scale=0.4]{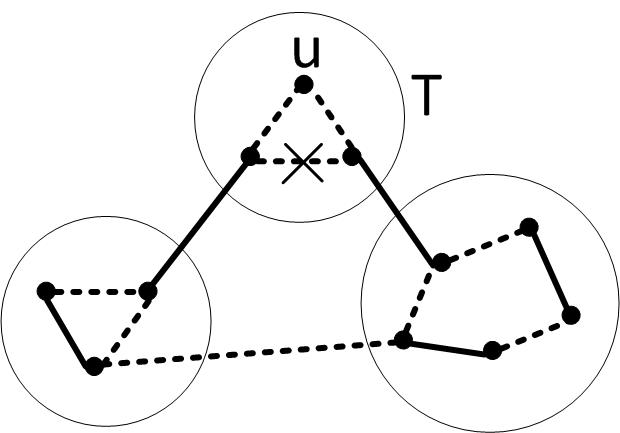}\\
An $(S\setminus u, \F)$ p.m.; $S$ is $(H,\F)$-factor-critical.  &   After deleting $e$, $S$ is not $(H,\F)$-factor-critical\\
& even though $S$ is factor-critical.
\end{tabular}
\caption{$(H,\F)$-factor-critical vs factor-critical}\label{fig:factor-critical}
\end{figure}

A feasible solution $\Pi$ to $D_\F(G,c)$ is an {\em $\fc{\F}$} dual,
if every $S\in \F$ is $(\Pi,\F)$-factor-critical, and $\Pi(T)>0$ for
every non-maximal set $T$ in $\F$. A family $\F\subseteq \O$ is called
a {\em critical family}, if $\F$ is laminar, and there exists an
$\fc{\F}$ dual solution.
This will be a significant notion: the set of cuts imposed in every iteration of
the cutting plane algorithm will be a critical family.
The following observation provides some context and motivation for these definitions.
\begin{prop}
Let $\F$ be the set of contracted sets at some stage of Edmonds' matching algorithm.
Then the corresponding dual solution $\Pi$ in the algorithm is an
$\fc{\F}$ dual.\proofbox
\end{prop}

We call $\Pi$ to be an {\em $\pfc{\F}$} dual, if $\Pi$ is a feasible solution to
\ref{prog:D-F}, and every $S\in \F$ such that $\Pi(S)>0$ is $(\Pi,\F)$-factor-critical.
Clearly, every $\fc{\F}$ dual is also an $\pfc{\F}$ dual, but the converse
is not true. The extremal dual optimal solutions
found in every iteration of Algorithm C-P-Matching will be $\pfc{\F}$, where $\F$ is the family of blossom inequalities imposed in that iteration.

The next lemma summarizes elementary properties of $\Pi$-critical
matchings.
\begin{lemma}\label{lem:critical-matching}
Let $\F$ be a laminar odd family, $\Pi$ be a feasible solution to
\ref{prog:D-F}, and $S\in \F$ be a $(\Pi,\F)$-factor-critical set.
For $u,v\in S$, let $M_u$, $M_v$ be the $\Pi$-critical-matchings for $u$, $v$
respectively.
\begin{enumerate}[(i)]
\item
For every $T\in \F$ such that $T\subsetneq S$,
\begin{align*}
|M_u\cap\delta(T)|=
\begin{cases}
1 &\mbox{ if } u\in S\setminus T,\\
0 &\mbox{ if } u\in T.
\end{cases}
\end{align*}
\item Assume the symmetric difference of $M_u$ and $M_v$ contains a
  cycle $C$. Then $M_u\Delta C$ is also a $\Pi$-critical
  matching for $u$. 
\end{enumerate}
\end{lemma}
\begin{proof}
{\em (i)} $M_u$ is a perfect matching of $S\setminus \{u\}$, hence for every
$T\subsetneq S$, 
\[
|M_u\cap \delta(T)|\equiv |T\setminus \{u\}|\pmod{2}.
\]
By definition of $M_u$, $|M_u\cap \delta(T)|\le 1$ for any $T\subsetneq
S$, $T\in \F$, implying the claim.

{\em (ii)} Let $M'=M_u\Delta C$. First observe that $u,v\not\in V(C)$. Hence $M'$ is a
perfect matching on $S\setminus \{u\}$ using only tight edges w.r.t.
$\Pi$. It remains to show that $|M'\cap \delta(T)|\le 1$ for every
$T\in \F$, $T\subsetneq S$. Let $\gamma_u$ and $\gamma_v$ denote the
number of edges in $C\cap \delta(T)$ belonging to $M_u$ and $M_v$,
respectively. Since these are critical matchings, we have $\gamma_u,\gamma_v\le 1$. On
the other hand, since $C$ is a cycle, $|C\cap \delta(T)|$ is even and hence
$\gamma_u+\gamma_v=|C\cap \delta(T)|$ is even. These imply that $\gamma_u=\gamma_v$.
The claim follows since $|M'\cap \delta(T)|=|M_u\cap \delta(T)|-\gamma_u+\gamma_v$.
\end{proof}

The following corollary shows that $(\Pi,\F)$-factor-critical property of a set implies that all sets contained inside it are also $(\Pi,\F)$-factor-critical.
\begin{corollary}\label{cor:factor-criticality-of-inside-sets}
Let $\F$ be a laminar family, $\Pi$ be a feasible solution to $D_{\F}(G,c)$, and $S\in \F$ be a $(\Pi,\F)$-factor-critical set. Then, every set $T\subseteq S,\ T\in \F$ is also $(\Pi,\F)$-factor-critical.
\end{corollary}
\begin{proof}
By definition, for each vertex $u\in S$, we have a matching $M_u$ supported on the tight edges of $\Pi$ such that (1) $M_u$ is a perfect matching on $S\setminus \{u\}$ and (2) $|M_u\cap\delta(U)|\le 1$ for all sets $U\subseteq S, U\in \F$.

Now, for any vertex $u\in T$, take $N_u=M_u\cap E[T]$. By Lemma \ref{lem:critical-matching}, we have that $|M_u\cap \delta(T)|=0$ and hence $N_u$ is a perfect matching on $T\setminus \{u\}$. Further, for each set $U\subseteq T, U\in \F$, we have that $|N_u\cap \delta(U)|\le |M_u\cap\delta(U)|\le 1$. Thus, $N_u$ is the required $\Pi$-critical-matching.
\end{proof}

The next claim is straightforward, since setting some components of a feasible solution of \ref{prog:D-F} to 0 also gives a feasible solution.
\begin{claim}\label{claim:downwards}
Let $\F$ be a critical family, and $\H\subseteq \F$ a downwards closed subfamily, i.e., if $S,T\in \F$, $S\subseteq T$ and $T\in\H$, then $S\in H$. Then $\H$ is also a critical family.
\end{claim}

The following {\em uniqueness} property is used to
guarantee the existence of a proper-half-integral solution in each step. We require
that the cost function $c:E\rightarrow \R$ satisfies:
\begin{align}
\mbox{For every critical family $\F$, \ref{prog:P-F} has a unique optimal
solution}.\tag{$\star$}\label{prop:uniqueness}
\end{align}
The next lemma shows that an arbitrary integer cost function can be
perturbed to satisfy this property. The proof of the lemma is presented in Section
\ref{sec:uniqueness}.
\begin{lemma}\label{lem:make-unique}
Let $c:E\rightarrow\Z$ be an integer cost function, and $\tilde c$ be its perturbation.
Then $\tilde c$ satisfies the uniqueness property (\ref{prop:uniqueness}).
\end{lemma}

\section{Analysis outline and proof of the main theorem}
The proof of our main theorem is established in two parts. In the first part,
we show that half-integrality of the intermediate primal optimum solutions is
guaranteed by the existence of an $\F$-positively-critical dual optimal solution to
\ref{prog:D-F}.

\begin{lemma}\label{lem:pof-uniqueness-implies-half-integrality}
Let $\F$ be a laminar odd family and assume \ref{prog:P-F} has a
unique optimal solution $x$. If $D_{\F}(G,c)$ has an $\pfc{\F}$ dual optimal solution, then
$x$ is proper-half-integral.
\end{lemma}
Lemma \ref{lem:pof-uniqueness-implies-half-integrality} is shown using a basic
contraction operation. Let $\Pi$ be an $\pfc{\F}$ dual optimal
solution for the  laminar odd family $\F$. Then contracting every set $S\in\F$ with
$\Pi(S)>0$  preserves primal and dual optimal solutions for the contracted graph and corresponding primal and dual LPs.
Lemma~\ref{lem:contract-factor-critical}.
Moreover, for a unique primal optimal solution $x$ to $P_{\F}(G,c)$, its image
$x'$ in the contracted graph is the
unique optimal solution; if $x'$ is proper-half-integral, then so is $x$.
Lemma~\ref{lem:pof-uniqueness-implies-half-integrality} then follows: we contract all
maximal sets $S\in\F$ with $\Pi(S)>0$. The image $x'$ of the unique optimal solution
$x$ is the unique optimal solution to the bipartite relaxation in the contracted graph,
and consequently, half-integral.

Such $\pfc{\F}$ dual optimal solutions are hence quite helpful, but their
existence is far from obvious.
We next show that if $\F$ is a critical family, then the extremal dual optimal
solutions found by the
algorithm are in fact $\pfc{\F}$ dual optimal solutions.

\begin{lemma}\label{lem:extremal-fit}
Suppose that in an iteration of Algorithm C-P-Matching, $\F$ is a critical family
with $\Gamma$ being an $\fc{\F}$ feasible solution to \ref{prog:D-F}. Then a $\Gamma$-extremal dual optimal
solution $\Pi$ is an $\pfc{\F}$ optimal solution to \ref{prog:D-F}. Moreover, the next set of
cuts $\H=\H'\cup\H''$ is a critical family with $\Pi$ being an $\fc{\H}$ dual to $D_{\H}(G,c)$.
\end{lemma}

Our goal then is to show that a critical family $\F$ always admits an $\pfc{\F}$ dual optimum; and
that every extremal dual solution satisfies this property.
We need a deeper understanding of the structure of dual optimal
solutions. Section~\ref{sec:dual-structure} is dedicated to this
analysis. Let $\Gamma$ be an $\fc{\F}$ dual solution, and $\Pi$ be an
arbitrary dual optimal solution to $D_{\F}(G,c)$.
Lemma~\ref{lem:consistency-main} shows the following relation between
$\Pi$ and $\Gamma$ inside sets $S\in\F$ that are tight for a
primal optimal solution $x$:
 Let $\Gamma_S(u)$ and $\Pi_S(u)$ denote the sum of the dual
values of sets containing $u$ that are strictly contained inside $S$
in solutions $\Gamma$ and $\Pi$ respectively, and let $\Delta=\max_{u\in
  S}(\Gamma_S(u)-\Pi_S(u))$. Then, every edge in $\supp(x)\cap
\delta(S)$ is incident to some node $u\in S$ such that $\Gamma_S(u)-\Pi_S(u)=\Delta$.
Also, if $S\in\F$ is both $\Gamma$- and
$\Pi$-factor-critical,
then $\Gamma$ and $\Pi$ are identical inside $S$ (Lemma~\ref{lem:factorcrit-ident}).

If $\Pi(S)>0$ but $S$ is not $\Pi$-factor-critical, the
above property (called {\em consistency} later) enables us to modify $\Pi$ by moving
towards $\Gamma$ inside $S$, and decreasing $\Pi(S)$ so that optimality is maintained.
Thus, we either get that $\Pi$ and $\Gamma$ are identical inside $S$ thereby making $S$
to be $\Pi$-factor-critical or $\Pi(S)=0$. A sequence of such operations converts an
arbitrary dual optimal solution to an $\pfc{\F}$ dual optimal one, leading to a
combinatorial procedure to obtain positively-critical dual optimal solutions
(Section~\ref{sec:pos-fact-crit}). Moreover, such operations decrease the secondary
objective value $h(\Pi,\Gamma)$ and thus show that every $\Gamma$-extremal dual optimum
is also an $\pfc{\F}$ dual optimum.

Lemmas~\ref{lem:pof-uniqueness-implies-half-integrality} and \ref{lem:extremal-fit}
together guarantee that the unique primal optimal solutions obtained during the
execution of the algorithm are proper-half-integral. In the second part of the proof of
Theorem~\ref{thm:main}, we show convergence by considering the number of odd cycles,
${\o}(x)$,  in the support of the current primal optimal solution $x$.

\begin{lemma}\label{lem:odd-cycles}
Assume the cost function $c$ satisfies (\ref{prop:uniqueness}). Then $\o(x)$ is
non-increasing during the execution of Algorithm C-P-Matching.
\end{lemma}

We observe that similar to
Lemma~\ref{lem:pof-uniqueness-implies-half-integrality}, the above
Lemma~\ref{lem:odd-cycles} is also true if we choose an arbitrary
$\pfc{\F}$ dual optimal solution $\Pi$ in each iteration of the
algorithm. To show that the number of cycles 
has to strictly decrease within a polynomial number of iterations, we need the more specific
choice of extremal duals.

\begin{lemma}\label{lem:strong-progress}
Assume the cost function $c$ satisfies (\ref{prop:uniqueness}) and that  ${\o}(x)$ does
not decrease between iterations $i$ and $j$, for some $i < j$. Let $\F_k$ be the set of blossom
inequalities imposed in the $k$'th iteration and $\H''_k=\F_k\setminus \F_{k-1}$ be the
subset of new inequalities in this iteration. Then,
\[
\bigcup_{k=i+1}^{j} \H''_k \subseteq \F_{j+1}.
\]
\end{lemma}

We prove this progress by coupling intermediate primal and
dual solutions with the solutions of a {\em Half-integral Matching} procedure that we design for this purpose. This procedure is 
a variation of Edmonds' primal-dual weighted matching algorithm and reveals the structure of the intermediate LP solutions. An extension of this procedure as described in \cite{our-matching-alg} leads to an algorithm for finding min-cost integral perfect matching. Unlike Edmonds' algorithm,
which maintains an integral matching and extends the matching to cover all vertices, the algorithm described in \cite{our-matching-alg} maintains a proper-half-integral perfect matching.

The main theorem can be proved using the above lemmas.
\begin{proof}[Proof of Theorem~\ref{thm:main}]
We use Algorithm C-P-Matching (Algorithm \ref{alg:main-alg}) for a perturbed cost
function. By Lemma~\ref{lem:make-unique}, this satisfies (\ref{prop:uniqueness}). Let
$i$ denote the index of the iteration. We prove by induction on $i$ that every
intermediate solution $x_{i}$ is proper-half-integral and thus {\em (i)} follows immediately by
the choice of the algorithm. The proper-half-integral property holds for the initial
solution $x_0$ by Proposition~\ref{prop:proper}. The induction step follows by Lemmas
\ref{lem:pof-uniqueness-implies-half-integrality} and \ref{lem:extremal-fit} and the
uniqueness property. Further, by Lemma \ref{lem:odd-cycles}, the number of odd cycles
in the support does not increase.

Assume the number of cycles in the $i$'th phase is $\ell$, and we have the same number
of odd cycles $\ell$ in a later iteration $j$. For $i\le k\le j$, the set
$\H''_k$ always contains $\ell$ cuts, and thus the number of
cuts added is at least $\ell (j-i)$. By
Lemma~\ref{lem:strong-progress}, all cuts in $\bigcup_{k=i+1}^{j} \H''_k $ are imposed
in
the family $\F_{j+1}$. Since $\F_{j+1}$ is a laminar odd family, it can contain
at most $n/2$ subsets, and therefore $j-i\le n/2\ell$. Consequently,
the number of cycles must decrease from $\ell$ to $\ell-1$ within
$n/2\ell$ iterations. Since $\o(x_0)\le n/3$, the number of iterations
is at most $O(n\log n)$.

Finally, we show that optimal solution returned by the algorithm using
$\tilde c$ is also optimal for the original cost function. Let $M$ be
the optimal matching returned by $\tilde c$, and assume for a
contradiction that there exists a different perfect matching $M'$
with $c(M')<c(M)$. Since $c$ is integral, it means $c(M')\le c(M)-1$. In the
perturbation, since $c(e)<\tilde c(e)$ for every $e\in E$, we have $c(M)<\tilde c(M)$,
and since $\sum_{e\in E}(\tilde c(e)-c(e))<1$, we have $\tilde c(M')< c(M')+1$.
This gives $\tilde c(M')<c(M')+1\le c(M)<\tilde c(M)$, a contradiction
to the optimality of $M$ for $\tilde c$.
\end{proof}

\section{Contractions and half-integrality}\label{sec:half-int}

We define an important contraction operation and derive
some fundamental properties. Let $\F$ be a laminar odd family, let $\Pi$ be a feasible
solution to $D_{\F}(G,c)$, and let $S\in \F$ be a $(\Pi,\F)$-factor-critical set. Let
us define
\[
\Pi_S(u):=\sum_{T\in \V\cup \F:T\subsetneq S, u\in T}\Pi(T)
\]
to be the total dual contribution of sets inside $S$ containing $u$.

By contracting $S$ w.r.t. $\Pi$, we mean the following: Let
$G'=(V',E')$ be the contracted graph on node set $V'=(V\setminus
S)\cup\{s\}$, $s$ representing the contraction of $S$. 
For each $u\in V$, we denote $u'$ to be the image of $u$, i.e., if $u\in S$, then $u'=s$, otherwise $u'=u$.
Let $\V'$ denote the set of one-element subsets of $V'$.
For a set $T\subseteq V$, let $T'$ denote its contracted image.
Let  $\F'$ be the set of nonsingular images of the sets of  $\F$,
that is, $T'\in \F'$ if $T\in \F$,
and $T\setminus S\neq\emptyset$. Let $E'$ contain all edges $uv\in
E$ with $u,v\notin S$ and for every edge $uv$ with $u\in S$, $v\in
V-S$ add an edge $u'v$.
(This may create parallel edges.) Let us define the image $\Pi'$ of $\Pi$ to be $\Pi'(T')=\Pi(T)$
for every $T'\in \V'\cup\F'$ and the image $x'$ of $x$ to be $x'(u'v')=x(uv)$ for every two vertices $u,v$ with at most one of them in $S$. Define
the new edge costs
\begin{align}
c'(u'v')=
\begin{cases}
c(uv)&\mbox{ if $uv\in E[V\setminus S]$},\notag\\
c(uv)-\Pi_S(u)&\mbox{ if $u\in S$, $v\in V\setminus S$}.
\end{cases}
\end{align}
We refer the reader to Figure \ref{fig:contraction} for an example of the contraction operation.
\begin{figure}[ht]\label{fig:contraction}
\centering
\includegraphics[scale=0.4]{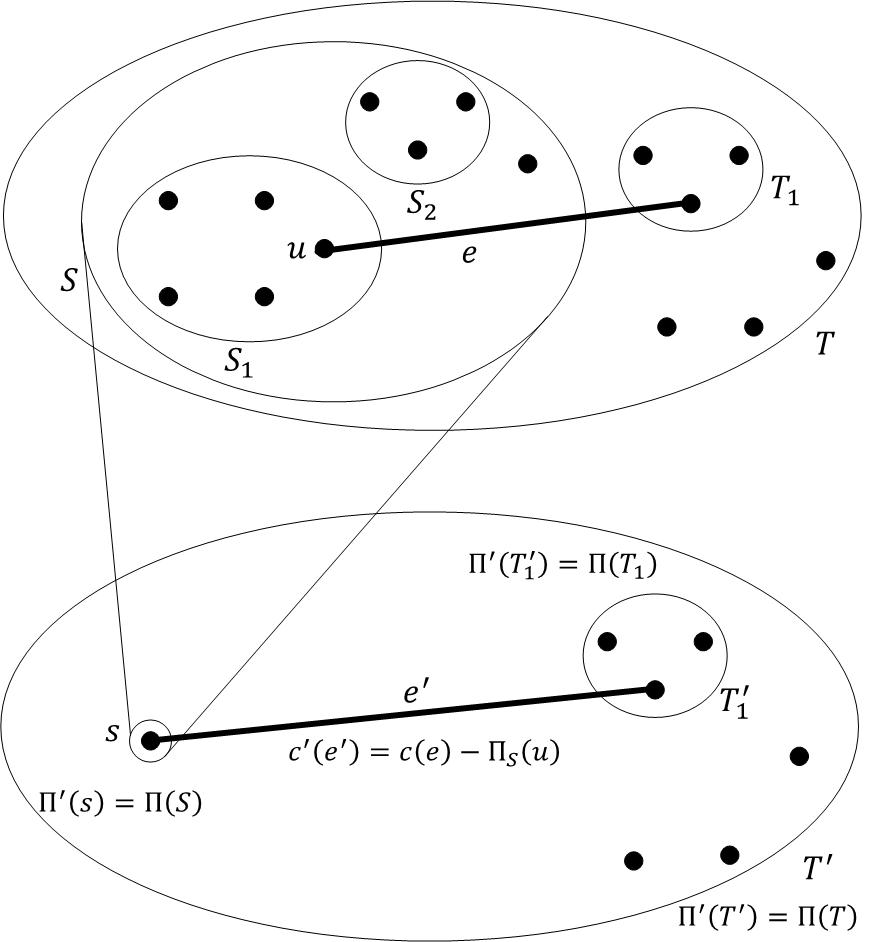}
\caption{Contraction operation: image of dual and new cost function}
\end{figure}

\begin{lemma}\label{lem:contract-factor-critical}
Let $\F$ be a laminar odd family, let $x$ be an optimal solution to
$P_{\F}(G,c)$, and let $\Pi$ be a feasible solution to $D_{\F}(G,c)$. Let $S\in\F$
be a $(\Pi,\F)$-factor-critical set, and let $G',c',\F'$ denote the
graph, costs and laminar family respectively obtained by contracting
$S$ w.r.t. $\Pi$; let $x',\Pi'$ be the images of $x,\Pi$
respectively. Then the following hold.
\begin{enumerate}[(i)]
\item $\Pi'$ is a feasible solution to $D_{\F'}(G',c')$. Furthermore, if a set $T\in
    \F$, $T\setminus S\neq \emptyset$ is $(\Pi,\F)$-factor-critical, then its image
    $T'$ is $(\Pi',\F')$-factor-critical.
\item Suppose $\Pi$ is an optimal solution to \ref{prog:D-F} and $x(\delta(S))=1$. Then
    $x'$ is an optimal solution to  $P_{\F'}(G',c')$ and $\Pi'$ is optimal to
    $D_{\F'}(G',c')$.
\item Suppose $x$ is the unique optimal solution to \ref{prog:P-F}, and $\Pi$ is an
  optimal solution to \ref{prog:D-F}. Then $x'$ is the unique optimal solution to
  $P_{\F'}(G',c')$. Moreover, $x'$ is proper-half-integral if and only if $x$ is
  proper-half-integral. If $x'$ is proper-half-integral, then $\o(x)=\o(x')$ and $\supp(x)\cap E[S]$ consists of a disjoint union of edges and an even path. 
	Further, assume $C'$ is an odd cycle in $\supp(x')$ and let $T$
  be the pre-image of $V(C')$ in $G$. Then, $\supp(x)\cap E[T]$ consists of an odd
  cycle and matching edges. 
  
\end{enumerate}
\end{lemma}

\begin{proof}
{\em (i)}
For feasibility, it is sufficient to verify
\[
\sum_{T'\in \V'\cup \F':u'v'\in \delta(T')}\Pi'(T')\le c'(u'v')\quad \forall u'v'\in
E'.
\]
If $u,v\neq s$, this is immediate from feasibility of $\Pi$ to $D_{\F}(G,c)$.
Consider an edge $sv'\in E(G')$. Let $uv$ be the pre-image of this edge.
\begin{align*}
\sum_{T'\in \V'\cup \F':sv'\in \delta(T')}\Pi'(T')
&=\Pi(S)+\sum_{T\in \F:uv\in \delta(T), T\setminus S\neq\emptyset}\Pi(T)
&\le c(uv)-\Pi_S(u)=c'(sv').
\end{align*}
We also observe that $u'v'$ is tight in $G'$ w.r.t $\Pi'$ if and only if the pre-image
$uv$ is tight in $G$ w.r.t $\Pi$.

Let $T\in \F$ be a $(\Pi,\F)$-factor-critical set with $T\setminus S\neq
\emptyset$. It is sufficient to verify that $T'$ is
$(\Pi',\F')$-factor-critical whenever
 $T$ contains $S$. Let $u'\in T'$ be the image of $u\in V$, and
 consider the image $M'$ of the $\Pi$-critical-matching $M_u$. Every edge in $M'$ is tight with respect to $\Pi'$. 
Further, $M'$ is a matching, since $|M_u\cap \delta(S)|\le 1$ must hold.
Let $Z'\subsetneq
 T'$, $Z'\in \F'$ and let $Z$ be the pre-image of $Z'$. If $u'\neq s$, then $|M'\cap
 \delta(Z')|=|M_u\cap \delta(Z)|\le 1$ and since $|M_u\cap \delta(S)|=1$ by
 Lemma~\ref{lem:critical-matching}, the matching $M'$ is a $(T'\setminus
 \{u'\},\F')$-perfect-matching. If $u'=s$, then $u\in S$. By
 Lemma~\ref{lem:critical-matching}, $M_u\cap \delta(S)=\emptyset$ and hence, $M'$
 misses $s$. Also, $|M_u\cap \delta(Z)|\le 1$ implies $|M'\cap\delta(Z')|\le 1$ and
 hence $M'$ is a $(T'\setminus\{s\},\F')$-perfect-matching.

{\em (ii)} Since $x(\delta(S))=1$, we have $x'(\delta(v))=1$ for every
$v\in V'$. It is straightforward to verify that
$x'(\delta(T'))\ge 1$ for every $T'\in \F'$ with equality if
$x(\delta(T))=1$. Thus, $x'$ is feasible to
$P_{\F'}(G',c')$. Optimality follows as $x'$ and $\Pi'$ satisfy
complementary slackness, using that the image of tight edges is tight,
as shown
by the argument for part {\em (i)}.

{\em(iii)} For uniqueness, consider an arbitrary optimal solution $y'$ to
$P_{\F'}(G',c')$.
We partition the edge set $\delta_{E'}(s)$ into subsets $J(u)$ for $u \in S$ as follows.
If $sv\in \delta_{E'}(s)$ and $u\in S$ is the unique node with $uv\in E$, then $sv\in \delta(u)$.
If there are multiple nodes $u_1,\ldots,u_t\in S$ with $u_iv\in E$, then let $u_i$ be an arbitrary node with an edge to $v$ that  minimizes $c(u_iv)-\Pi_S(u_i)$ and let $sv\in J(u_i)$.
Define $\alpha_u=y'(J(u))$. It is straightforward by the contraction that $\sum_{u\in S} \alpha_u=y'(\delta(s))=1$.
Let $M_u$ be the $\Pi$-critical matching for $u$ in $S$.
Take $w=\sum_{u\in S}\alpha_u M_u$ and
\begin{equation*}
y(uv)=
\begin{cases}
y'(u'v') & \mbox{ if $uv\in E\setminus E[S]$},\\
w(uv) & \mbox{ if $uv\in E[S]$}.
\end{cases}
\end{equation*}
Now it is easy to verify that $y$ is a feasible solution to $P_{\F}(G,c)$. Moreover, $y'$ satisfies complementary
slackness with $\Pi'$, which is a dual optimal solution for $D_{\F}(G',c')$, and therefore $y$ satisfies complementary slackness with $\Pi$. \vnote{Does this need further argument?} \knote{This would blow-up the proof without any interesting insights.} 
Hence, $y$ is an optimal solution to $P_{\F}(G,c)$ and thus by uniqueness,
$y=x$. Consequently, $y'=x'$.

The above argument also shows that $x$ must be identical to $w$ inside $S$.
Suppose $x'$ is proper-half-integral. First, assume $s$ is covered by a matching edge
in $x'$. Then $\alpha_u=1$ for some $u\in S$ and $\alpha_v=0$ for every $v\neq u$.
Consequently, $w=M_u$ is a perfect matching on $S-u$. Next, assume $s$ is incident to
an odd cycle in $x'$. Then $\alpha_{u_1}=\alpha_{u_2}=1/2$ for some nodes $u_1,u_2\in
S$, and $w=\frac12 (M_{u_1}+M_{u_2})$. The uniqueness of $x$ implies the uniqueness of
both $M_{u_1}$ and $M_{u_2}$. Then by Lemma~\ref{lem:critical-matching}(ii), the
symmetric difference of $M_{u_1}$ and $M_{u_2}$ does not contain any even cycles. Hence,
$\supp(w)$ contains an even path between $u_1$ and $u_2$, and some matching edges.
Consequently, $x$ is proper-half-integral.
The above argument immediately shows the following.
\begin{claim}\label{cl:blow-cycle}
Let $C'$ be an odd (even) cycle such that $x'(e)=1/2$ for every $e\in C'$ in
$\supp(x')$ and let $T$ be the pre-image of the set $V(C')$ in $G$. Then, $\supp(x)\cap
E[T]$ consists of an odd (even) cycle $C$ and a (possibly empty) set $M$ of edges such
that $x(e)=1/2\ \forall\ e\in C$ and $x(e)=1\ \forall\ e\in M$.
\proofbox
\end{claim}

Next, we
 prove that if $x$ is proper-half-integral, then so is $x'$. It is clear that $x'$
 being the image of $x$ is half-integral. If $x'$ is not proper-half-integral, then
 $\supp(x')$ contains an even $1/2$-cycle, and thus by Claim~\ref{cl:blow-cycle},
 $\supp(x)$ must also contain an even cycle, contradicting that it was
 proper-half-integral.
 
 The above arguments also show $\supp(x)\cap E[S]$ consists of a disjoint union of edges and an even path, $\o(x)=\o(x')$, and finally, if $C'$ is an odd cycle in $\supp(x')$, then Claim~\ref{cl:blow-cycle} provides the required structure for $x$ inside $T$.
%
\end{proof}

Iteratively applying the lemma from the innermost to the outermost sets in $\F$, we obtain the following corollary.
\begin{corollary}\label{cor:contract-bipartite-opt}
Assume $x$ is the  optimal solution to \ref{prog:P-F} and there
exists an $\pfc{\F}$ dual optimal solution $\Pi$. Let $\hat G,\hat c$ be the graph, and cost
obtained by contracting all maximal sets $S\in\F$ with
$\Pi(S)>0$ w.r.t. $\Pi$, and let $\hat x$ be the image of $x$ in $\hat
G$.
\begin{enumerate}[(i)]
\item
$\hat x$ and $\hat \Pi$ are the optimal solutions to the
bipartite relaxation $P_0(\hat G,\hat c)$ and $D_0(\hat G,\hat c)$
respectively.
\item
If $x$ is the unique optimal solution to $P_{\F}(G,c)$, then $\hat x$ is the
unique optimal solution to $P_{0}(\hat G,\hat c)$. If
  $\hat x$ is proper-half-integral, then $x$ is also proper-half-integral.
\end{enumerate}
\end{corollary}

\begin{proof}[Proof of Lemma  \ref{lem:pof-uniqueness-implies-half-integrality}]
Let $\Pi$ be an $\pfc{\F}$ dual optimal solution, and let $x$ be the unique
optimal solution to $P_{\F}(G,c)$. Contract all maximal sets $S\in \F$
with $\Pi(S)>0$, obtaining the graph $\hat G$ and cost $\hat c$. Let $\hat x$ be the
image of $x$ in $\hat G$.
By Corollary~\ref{cor:contract-bipartite-opt}(ii), $\hat x$ is unique
optimal solution to $P_0(\hat G,\hat c)$. By
Proposition~\ref{prop:proper}, $\hat x$ is proper-half-integral and
hence by  Corollary~\ref{cor:contract-bipartite-opt}(ii), $x$ is also
proper-half-integral.
\end{proof}

\section{Structure of dual solutions}\label{sec:dual-structure}
In this section, we derive two properties of positively-critical dual optimal solutions: {\em(1)} an optimal solution $\Psi$ to $D_{\F}(G,c)$ can be transformed into an $\pfc{\F}$ dual
optimal solution if $\F$ is a critical family (Section \ref{sec:pos-fact-crit}) and {\em(2)} a
$\Gamma$-extremal dual optimal solution to $D_{\F}(G,c)$ as obtained in the algorithm
is also an $\pfc{\F}$ dual optimal solution (Section
\ref{sec:extremal-dual-solutions}). In Section \ref{sec:consistency-of-duals}, we first
show some lemmas characterizing arbitrary dual optimal solutions.

\subsection{Consistency of  dual
  solutions}\label{sec:consistency-of-duals}
Assume $\F\subseteq \O$ is a critical family, with $\Pi$ being an
$\fc{\F}$ dual solution, and let $\Psi$ be an arbitrary dual optimal
solution to \ref{prog:D-F}. Note that optimality of $\Pi$ is not assumed.
Let $x$ be an optimal solution to \ref{prog:P-F}; we do not make the
uniqueness assumption (\ref{prop:uniqueness}) in this section.
We shall describe structural properties of $\Psi$ compared to $\Pi$;
in particular, we show
that if we contract a $\Pi$-factor-critical set $S$, the  images of
$x$ and $\Psi$ will be primal and dual optimal solutions in the
contracted graph.

Consider a set $S\in \F$. We say that the dual solutions $\Pi$ and $\Psi$ are {\em
identical} inside $S$, if $\Pi(T)=\Psi(T)$ for every set $T\subsetneq S$, $T\in
\F\cup\V$.
We defined $\Pi_S(u)$ in the previous section; we
also use this notation for $\Psi$, namely, let
$\Psi_S(u):=\sum_{T\in \V\cup \F:T\subsetneq S, u\in T}\Psi(T)$ for
$u\in S$.
Let us now define
\[
\Delta_{\Pi,\Psi}(S):=\max _{u\in S} \left(\Pi_S(u)-\Psi_S(u)\right).
\]
We say that $\Psi$ is {\em consistent} with $\Pi$ inside $S$, if
$\Pi_S(u)-\Psi_S(u)=\Delta_{\Pi,\Psi}(S)$ holds for every $u\in S$ that is incident to
an edge $uv\in \delta(S)\cap \supp(x)$. The main goal of this section is to prove the
following lemma.

\begin{lemma}\label{lem:consistency-main}
Let $\F\subseteq \O$, $\Pi$ be a feasible solution to $D_{\F}(G,c)$, and let $S\in \F$ such that $S$ is $(\Pi,\F)$-factor-critical and $\Pi(T)>0$ for every subset $T\subsetneq S, T\in \F$. 
Let $\Psi$ be an optimal solution to \ref{prog:D-F}, and $x$ be an
optimal solution to \ref{prog:P-F} with $x(\delta(S))=1$. 
Then  $\Psi$ is consistent with $\Pi$ inside $S$.
Further, $\Delta_{\Pi,\Psi}(S)\ge 0$
for all such sets $S$.
\end{lemma}

Consistency is important as it enables us to preserve optimality when contracting a set
$S\in \F$ w.r.t. $\Pi$. Assume $\Psi$ is consistent with $\Pi$ inside $S$, and
$x(\delta(S))=1$. Let us  contract $S$ w.r.t. $\Pi$ to obtain $G'$  and $c'$ as defined
in Section~\ref{sec:half-int}.
Define
\begin{align*}
\Psi'(T')&=
\begin{cases}
\Psi(T) & \mbox{ if } T'\in (\F'\cup\V')\setminus\{s\},\\
\Psi(S)-\Delta_{\Pi,\Psi}(S) & \mbox{ if } T'=\{s\} 
\end{cases}
\end{align*}

\begin{lemma}\label{lem:contract-consistent}
Let $\F\subseteq \O$, $\Pi$ be a feasible solution to $D_{\F}(G,c)$, and let $S\in \F$ such that $S$ is $(\Pi,\F)$-factor-critical and $\Pi(T)>0$ for every subset $T\subsetneq S, T\in \F$. 
Let $\Psi$ be an optimal solution to \ref{prog:D-F}, and $x$ be an
optimal solution to \ref{prog:P-F} with $x(\delta(S))=1$. Suppose that
$\Psi$ is consistent with $\Pi$ inside $S$.
Let $G',c',\F'$ denote the graph, costs and laminar
family obtained by contraction. Then the image $x'$ of $x$ is an optimal solution to
$P_{\F'}(G',c')$, and $\Psi'$ (as defined above) is an optimal solution to $D_{\F'}(G',c')$.
\end{lemma}
\begin{proof}
Feasibility of $x'$ follows as in the proof of
Lemma~\ref{lem:contract-factor-critical}(ii).
For the feasibility of $\Psi'$, we have to verify  $\sum_{T'\in \V'\cup \F':uv\in
\delta(T')}\Psi'(T')\le c'(uv)$ for every edge $uv\in E(G')$.
This follows immediately for every edge $uv$ such that $u,v\neq s$
since $\Psi$ is a feasible solution for $D_{\F}(G,c)$. Consider an
edge $uv\in E(G)$, $u\in S$. Let $sv\in E(G')$ be the image of $uv$ in
$G'$, and
let $\Delta=\Delta_{\Pi,\Psi}(S)$.
\begin{align*}
c(uv)
&\ge \sum_{T\in \V\cup \F:uv\in \delta(T)}\Psi(T)\\
&=\Psi_S(u)+\Psi(S)+\sum_{T\in \F:uv\in \delta(T), T\setminus S\neq\emptyset}\Psi(T)\\
&=\Psi_S(u)+\Delta+\sum_{T'\in \V'\cup \F':sv\in \delta(T')}\Psi'(T').
\end{align*}
In the last equality, we used the definition $\Psi'(s)=\Psi(S)-\Delta$.
Therefore, using $\Pi_S(u)\le \Psi_S(u)+\Delta$, we obtain
\begin{align}
\sum_{T'\in \V'\cup \F':sv\in \delta(T')}\Psi'(T')&\le c(uv)-\Psi_S(u)-\Delta\le
c(uv)-\Pi_S(u)= c'(uv). \label{tightness}
\end{align}
Thus, $\Psi'$ is a feasible solution to $D_{\F'}(G',c')$.
To show optimality, we verify complementary slackness for $x'$ and $\Psi'$.
If $x'(uv)>0$
for $u,v\neq s$, then $x(uv)>0$. Thus, the tightness of the constraint
for $uv$ w.r.t. $\Psi'$ in $D_{\F'}(G',c')$ follows from the tightness of the
constraint w.r.t. $\Psi$ in $D_{\F}(G,c)$. Suppose $x'(sv)>0$ for an edge $sv\in
E(G')$. Let $uv\in E(G)$ be the pre-image of $sv$ for some $u\in
S$. Then the tightness of the constraint follows since both the
inequalities in (\ref{tightness}) are tight -- the first inequality is
tight since $uv$ is tight w.r.t. $\Psi$, and the second is tight since
$\Pi_S(u)-\Psi_S(u)=\Delta(S)$ by the consistency property. Finally, if $\Psi'(T')>0$
for some $T'\in \F'$, then $\Psi(T)>0$ and hence $x(\delta(T))=1$, implying
$x'(\delta(T'))=1$.
\end{proof}

\begin{lemma}\label{lem:tight-inside}
Let $\F\subseteq \O$, $\Pi$ be a feasible solution to $D_{\F}(G,c)$. Let $S\in \F$ such that $S$ is $(\Pi,\F)$-factor-critical, and $\Pi(T)>0$ for every subset $T\subsetneq S, T\in \F$.
Let $x$ be an  optimal solution to \ref{prog:P-F}. If  $x(\delta(S))=1$, 
then all edges in $\supp(x)\cap E[S]$ are tight w.r.t. $\Pi$ and $x(\delta(T))=1$ for every $T\subsetneq S$, $T\in \F$.
\end{lemma}
\begin{proof}
Let $\alpha_u=x(\delta(u,V\setminus S))$ for each $u\in S$, and for each $T\subseteq
S$, $T\in \F$, let $\alpha(T)=\sum_{u\in T}\alpha_u=x(\delta(T,V\setminus S))$. Note
that $\alpha(S)=x(\delta(S))=1$. Let us consider the following pair of linear
programs.

\begin{multicols}{2}
\noindent
\begin{align}
\min& \sum_{uv\in E[S]} c(uv) z(uv) \tag{$P_{\F}[S]$}\label{prog:P-S-F}\\
z(\delta(u))&=1-\alpha_u\quad\forall u\in S\notag\\
z(\delta(T))&\ge 1-\alpha(T)\quad \forall T\subsetneq S, T\in {\cal F}\notag\\
z(uv)&\ge0\ \forall\ uv\in E[S]\notag
\end{align}
\begin{align}
\max \sum_{T\subsetneq S, T\in\V\cup{\cal F}}&(1-\alpha(T))\Gamma(T) \tag{$D_{\cal
F}[S]$}\label{prog:D-S-F}\\\
\sum_{\substack{T\subsetneq S, T\in\V\cup{\cal F}\\uv\in \delta(T)}}\Gamma(T) &\le
c(uv) \quad \forall
uv\in E[S]\notag \\
\Gamma(Z)&\ge0\quad \forall Z\subsetneq T, Z\in F\notag
\end{align}
\end{multicols}

For a feasible solution $z$ to \ref{prog:P-S-F}, let $x^z$ denote the solution obtained
by replacing $x(uv)$ by $z(uv)$ for edges $uv$ inside $S$, that is,
\begin{align*}
{x}^z(e)&=
\begin{cases}
x(e) & \mbox{ if } e\in \delta(S)\cup E[V\setminus S],\\
z(e) & \mbox{ if } e\in E[S].
\end{cases}
\end{align*}

\begin{claim}\label{claim:P-S-F}
The restriction of $x$ inside $S$ is feasible
to \ref{prog:P-S-F}, and for
every feasible solution $z$ to \ref{prog:P-S-F}, $x^z$ is a
feasible solution to \ref{prog:P-F}. Consequently, $z$ is an
optimal solution to \ref{prog:P-S-F} if and only if $x^z$ is an optimal solution to
\ref{prog:P-F}.
\end{claim}
\begin{proof}
The first part is obvious.
For feasibility of $x^z$, if $u\notin S$ then $x^z(u)=x(u)=1$. If
$u\in S$, then $x^z(u)=z(u)+x(\delta(u,V\setminus
S))=1-\alpha_u+\alpha_u=1$. Similarly, if $T\in \F$, $T\setminus
S\neq\emptyset $, then $x^z(\delta(T))=x(T)\ge 1$. If $T\subseteq S$,
then $x^z(\delta(T))=z(\delta(T))+x(\delta(T,V\setminus S))\ge
1-\alpha(T)+\alpha(T)=1$.

Optimality follows since $c^Tx^z=\sum_{uv\in
  E[S]}c(uv)z(uv)+\sum_{uv\in E\setminus E[S]}c(uv)x(uv)$.
\end{proof}

\begin{claim}\label{claim:w-opt}
Let $\bar\Pi$ denote the restriction of $\Pi$ inside $S$, that is,
$\bar\Pi(T)=\Pi(T)$ for every $T\in\V\cup\F$, $T\subsetneq S$.
Then $\bar\Pi$ is an
optimal solution to \ref{prog:D-S-F}. 
\end{claim}
\begin{proof}
Since $S\in \F$ is $(\Pi,\F)$-factor-critical, 
we have a $\Pi$-critical-matching $M_u$
inside $S$ for each $u\in S$. Let $z=\sum_{u\in S}\alpha_uM_u$.
The claim follows by showing that $z$ is  feasible to
\ref{prog:P-S-F} and that $z$ and $\bar\Pi$ satisfy complementary slackness.

The degree constraint $z(\delta(u))=1-\alpha_u$ is straightforward. By
Lemma~\ref{lem:critical-matching}(i), if $T\subsetneq S$, $T\in\F$, then
$z(\delta(T))=\sum_{u\in S\setminus T}\alpha_u=1-\alpha(T)$.
The feasibility of $\Pi$ to \ref{prog:D-F} immediately shows
feasibility of $\bar\Pi$ to \ref{prog:D-S-F}.

Complementary slackness also follows since by definition, all $M_u$'s use only tight
edges w.r.t. $\Pi$ (equivalently, w.r.t. $\bar\Pi$). Also, for every odd set
$T\subsetneq S$, $T\in\F$, we have that $z(\delta(T))=1-\alpha(T)$ as verified above.
Thus, all odd set constraints are tight in the primal.
\end{proof}

By Claim~\ref{claim:P-S-F}, the solution obtained by restricting $x$ to $E[S]$ must be
optimal to
\ref{prog:P-S-F} and thus satisfies complementary slackness with $\bar\Pi$.
Consequently, every edge in $E[S]\cap \supp(x)$ must be
tight w.r.t. $\bar\Pi$, and equivalently, w.r.t. $\Pi$.
By the statement of the Lemma, every set
$T\subsetneq
S$, $T\in \F$ satisfies $\bar \Pi(T)=\Pi(T)>0$. Thus, complementary slackness gives
$x(\delta(T))=1$.

\end{proof}

We need one more claim to prove Lemma~\ref{lem:consistency-main}.
\begin{claim}\label{claim:a-minus-greater-than-a-plus}
Let $S\in\F$ be an inclusionwise minimal set of $\F$.
Let $\Lambda$ and $\Gamma$ be feasible solutions to \ref{prog:D-F}, and suppose $S$ is
$(\Lambda,\F)$-factor-critical. Then, 
\[
\Delta_{\Lambda,\Gamma}(S):=\max_{u \in S} (\Lambda_S(u) - \Gamma_S(u)) = \max _{u\in S}
|\Lambda_S(u)-\Gamma_S(u)|.
\]
Further, if $\Delta_{\Lambda,\Gamma}(S)>0$, define
 \begin{align*}
A^+&:=\{u\in S:\Gamma(u)=\Lambda(u)+\Delta_{\Lambda,\Gamma}(S)\},\\
A^-&:=\{u\in S:\Gamma(u)=\Lambda(u)-\Delta_{\Lambda,\Gamma}(S)\}.
\end{align*}
Then $|A^-|>|A^+|$.
\end{claim}
\begin{proof}
Let $\Delta=\max _{u\in S} |\Lambda_S(u)-\Gamma_S(u)|$; note that $\Delta\ge \Delta_{\Lambda,\Gamma}(S)$ by definition.
If $\Delta=0$, then $\Delta_{\Lambda,\Gamma}(S)=0$ also follows, and thus the claim holds. In the rest of the proof, we  assume $\Delta>0$.
Let us define the
sets $A^-$ and $A^+$ with $\Delta$ instead of $\Delta_{\Lambda,\Gamma}(S)$.
Since $S$ is $(\Lambda,\F)$-factor-critical, for every $a\in S$, there exists an
$(S\setminus\{a\},\F)$ perfect matching $M_a$ using only tight edges w.r.t. $\Lambda$,
i.e., $M_a\subseteq \{uv:\Lambda(u)+\Lambda(v)=c(uv)\}$ by the minimality of $S$.
Further, by feasibility of $\Gamma$, we have $\Gamma(u)+\Gamma(v)\le c(uv)$ on every
$uv\in M_a$. Thus, if $u\in A^+$, then $v\in A^-$ for every $uv\in M_a$. Since
$\Delta>0$, we have $A^+\cup  A^-\neq \emptyset$ and therefore $A^-\neq \emptyset$, and
consequently, $\Delta=\Delta_{\Lambda,\Gamma}(S)$.
Now pick $a\in A^-$ and consider $M_a$. This perfect matching $M_a$ matches each node
in $A^+$ to a node in $A^-$. Thus, $|A^-|>|A^+|$.
\end{proof}

\begin{proof}[Proof of Lemma~\ref{lem:consistency-main}]
We prove by induction on $|V|$, and subject to that, on
$|S|$. Let us define $\Delta:=\Delta_{\Pi,\Psi}(S)$.
By the statement of the lemma, we have that $S$ is $(\Pi,\F)$-factor-critical. 


First, consider the case when $S$ is an inclusion-wise minimal
set. Then, $\Pi_S(u)=\Pi(u)$, $\Psi_S(u)=\Psi(u)$ for every $u\in
S$. By Claim~\ref{claim:a-minus-greater-than-a-plus}, we have 
$\Delta\ge 0$. We are done if $\Delta=0$.  Otherwise, define the
sets $A^-$ and $A^+$ as in the claim using $\Delta_{\Pi,\Psi}(S)$.

Now consider an edge $uv\in E[S]\cap \supp(x)$. By complementary slackness, we have
$\Psi(u)+\Psi(v)=c(uv)$. By dual feasibility, we have $\Pi(u)+\Pi(v)\le c(uv)$. Hence,
if $u\in A^-$, then $v\in A^+$. Consequently, we have
\begin{align*}
|A^-|&=\sum_{u\in A^-}x(\delta(u))= x(\delta(A^-,V\setminus S))+x(\delta(A^-,A^+))\\
&\le 1+\sum_{u\in A^+}x(\delta(u))=1+|A^+|\le |A^-|.
\end{align*}
Thus, we must have equality throughout, implying
$x(\delta(A^-,V\setminus S))=1$.  This precisely means that $\Psi$ is consistent
with $\Pi$ inside $S$.

Next, let $S$ be a non-minimal set. 
Let $T\in \F$ be a maximal set
strictly contained in $S$. By Corollary \ref{cor:factor-criticality-of-inside-sets}, we know that $T$ is also $(\Pi,\F)$-factor-critical.
By Lemma~\ref{lem:tight-inside}, $x(\delta(T))=1$, therefore the
inductional claim holds for $T$: $\Psi$ is consistent with $\Pi$
inside $T$, and $\Delta(T)=\Delta_{\Pi,\Psi}(T)\ge 0$.

We contract $T$ w.r.t. $\Pi$ and use  Lemma~\ref{lem:contract-consistent}. Let the
image of the solutions $x$, $\Pi$, and $\Psi$ be $x'$,
$\Pi'$ and $\Psi'$ respectively and the resulting graph be $G'$ with
cost function $c'$.
Then $x'$ and $\Psi'$ are  optimal solutions to $P_{\F'}(G',c')$ and
to $D_{\F'}(G',c')$ respectively, and by
Lemma~\ref{lem:contract-factor-critical}(i),
$\Pi'$ is an $\fc{\F'}$ dual. Let $t$ be the image of $T$ by the contraction.
Now, consider the image $S'$ of $S$ in $G'$. Since $G'$ is a smaller graph, it
satisfies the induction hypothesis. Let $\Delta'=\Delta_{\Pi',\Psi'}(S')$ in $G'$. By
induction hypothesis, $\Delta'\ge 0$. The following claim verifies consistency inside
$S$ and thus completes the proof.
\end{proof}

\begin{claim}\label{claim:Delta}
For every $u\in S$,
$\Pi_S(u)-\Psi_S(u)\le\Pi'_{S'}(u')-\Psi'_{S'}(u')$,
 and equality holds if there exists an edge $uv\in \delta(S)\cap \supp(x)$.
Consequently, $\Delta'=\Delta$.
\end{claim}
\begin{proof}
Let $u'$ denote the image of $u$. If $u'\neq t$, then  $\Pi'_{S'}(u')=\Pi_S(u),
\Psi'_{S'}(u')=\Psi_S(u)$ and therefore,
$\Pi_{S}(u)-\Psi_{S}(u)=\Pi'_{S'}(u')-\Psi'_{S'}(u')$. Assume $u'=t$, that is, $u\in
T$. Then $\Pi_S(u)=\Pi_T(u)+\Pi(T)$, $\Psi_S(u)=\Psi_T(u)+\Psi(T)$ by the maximal
choice of $T$, and therefore
\begin{align}
\Pi_S(u)-\Psi_S(u)
&=\Pi_T(u)-\Psi_T(u)+\Pi(T)-\Psi(T)\notag\\
&\le\Delta(T)+\Pi(T)-\Psi(T)\notag\\
&=\Pi'(t)-\Psi'(t) \quad \quad \text{(Since $\Pi'(t)=\Pi(T)$,
$\Psi'(t)=\Psi(T)-\Delta(T)$)}\notag\\
&=\Pi'_{S'}(t)-\Psi'_{S'}(t).\label{eq:u-in-T}
\end{align}
Assume now that there exists a $uv\in \delta(S)\cap \supp(x)$.
If $u\in T$, then using the consistency inside $T$, we get
$\Pi_T(u)-\Psi_T(u)=\Delta(T)$, and therefore (\ref{eq:u-in-T}) gives
$\Pi_S(u)-\Psi_S(u)=\Pi'_{S'}(t)-\Psi'_{S'}(t)=\Delta'$.
\end{proof}

Claim~\ref{claim:a-minus-greater-than-a-plus} can also be used to derive the following
important property.
\begin{lemma}\label{lem:factorcrit-ident}
Given a laminar odd family $\F\subset \O$, let $\Lambda$ and $\Gamma$ be two
dual feasible solutions to $D_{\F}(G,c)$. If a subset $S\in \F$ is both
$(\Lambda,\F)$-factor-critical and $(\Gamma,\F)$-factor-critical, then $\Lambda$ and
$\Gamma$ are identical inside $S$.
\end{lemma}
\begin{proof}
Consider a graph $G=(V,E)$ with $|V|$ minimal, where the claim does not hold for some
set $S$. Also, choose  $S$ to be the smallest counterexample in this graph. First,
assume $S\in \F$ is a minimal set. Then consider
Claim~\ref{claim:a-minus-greater-than-a-plus} for $\Lambda$ and $\Gamma$
and also by changing their roles, for $\Gamma$ and $\Lambda$. If
$\Lambda$ and $\Gamma$ are not identical inside $S$, then
$\Delta=\max_{u\in S}|\Lambda_S(u)-\Gamma_S(u)|>0$. The sets $A^-$
and $A^+$ for $\Lambda$ and $\Gamma$ become $A^+$ and $A^-$ for $\Gamma$ and
$\Lambda$. Then $|A^-|>|A^+|>|A^-|$, a contradiction.

Suppose now $S$ contains $T\in \F$. It is straightforward by definition that $T$ is also
$(\Lambda,\F)$-factor-critical and $(\Gamma,\F)$-factor-critical.  Thus,
by the minimal choice of the counterexample $S$, we have that $\Lambda$ and $\Gamma$
are identical inside $T$. Now, contract the set $T$ w.r.t. $\Lambda$, or equivalently,
w.r.t. $\Gamma$. Let $\Lambda'$, $\Gamma'$ denote the contracted solutions in $G'$, and
let $\F'$ be the contraction of $\F$. Then, by
Lemma~\ref{lem:contract-factor-critical}(i), these two solutions are feasible to
$D_{\F'}(G',c')$, and $S'$ is both $\Lambda'$-factor-critical and
$\Gamma'$-factor-critical. Now, $\Lambda'$ and $\Gamma'$ are not identical inside $S'$,
contradicting the minimal choice of $G$ and $S$.
\end{proof}

\subsection{Finding a positively-critical dual optimal
solution}\label{sec:pos-fact-crit}
Let $\F\subseteq \O$ be a critical family with $\Pi$ being an
$\fc{\F}$ dual. Let $\Psi$ be a dual optimal solution to
$D_{\F}(G,c)$. We present Algorithm \ref{alg:positive-alg} that modifies  $\Psi$ to an $\pfc{\F}$ dual optimal solution.
 The correctness of the algorithm follows by showing that in every iteration, the modified solution $\bar \Psi$ is also dual optimal, and it is ``closer'' to $\Pi$.

\begin{algorithm}
\caption{Algorithm Positively-critical-dual-opt}\label{alg:positive-alg}

\noindent Input: An optimal solution $\Psi$ to $D_{\F}(G,c)$ and a $\fc{\F}$ dual solution $\Pi$ to $D_{\F}(G,c)$\\
\noindent Output: An $\pfc{\F}$ dual optimal solution to $D_{\F}(G,c)$

\begin{enumerate}
\item {\bf Repeat} while $\Psi$ is not $\pfc{\F}$ dual.
\begin{enumerate}
\item Choose a maximal set $S\in\F$ with $\Psi(S)>0$, such that $\Pi$ and $\Psi$ are
    not identical inside $S$.
\item Set $\Delta:=\Delta_{\Pi,\Psi}(S)$.
\item Let  $\lambda:=\min\{1,\Psi(S)/\Delta\}$ if $\Delta>0$ and $\lambda:=1$ if
    $\Delta=0$.
\item Replace $\Psi$ by the following $\bar \Psi$.
\begin{align}
\bar \Psi(T)&:=
\begin{cases}
&(1-\lambda)\Psi(T)+\lambda \Pi(T) \mbox{ if }T\subsetneq S,\\
& \Psi(S)-\Delta\lambda \mbox{ if }T=S,\label{eq:q-mod}\\
& \Psi(T)\mbox{ otherwise }.
\end{cases}
\end{align}
\end{enumerate}
\item {\bf Return} $\Psi$.
\end{enumerate}
\end{algorithm}

\begin{lemma}\label{lem:bar-q}
Let $\F\subseteq \O$ be a critical family with $\Pi$ being an $\fc{\F}$ dual and let
$\Psi$ be a dual optimal solution to $D_{\F}(G,c)$. Suppose we consider a maximal set
$S$ such that $\Pi$ and $\Psi$ are not identical inside $S$, and $\Psi(S)>0$. Define
$\lambda=\min\{1,\Psi(S)/\Delta_{\Pi,\Psi}(S)\}$ if $\Delta_{\Pi,\Psi}(S)>0$ and
$\lambda=1$ if $\Delta_{\Pi,\Psi}(S)=0$ and set $\bar \Psi$ as in (\ref{eq:q-mod}).
Then, $\bar \Psi$ is also a dual optimal solution to \ref{prog:D-F}, and either $\bar
\Psi(S)=0$ or $\Pi$ and $\bar{\Psi}$ are identical inside $S$.
\end{lemma}
\begin{proof}
Let $x$ be an optimal solution to \ref{prog:P-F}. Since $\Psi(S)>0$,
we have $x(\delta(S))=1$ and by Lemma \ref{lem:consistency-main}, we
have
$\Delta=\Delta_{\Pi,\Psi}(S)\ge 0$. Now, the second conclusion is
immediate from definition: if $\lambda=1$, then we have that $\Pi$ and
$\bar \Psi$ are identical inside $S$; if $\lambda<1$, then we have
$\bar \Psi(S)=0$. For optimality, we show feasibility and verify the
primal-dual slackness conditions.

The solution $\bar \Psi$ might have positive components on some sets
$T\subsetneq S, T\in \F$ where $\Psi(T)=0$  (but $\Pi(T)>0$). However, $x(\delta(T))=1$
for all sets $T\subsetneq S, T\in \F$ by Lemma~\ref{lem:tight-inside}, since $x(\delta(S))=1$ by complementary slackness between $x$ and $\Psi$. The choice of $\lambda$ also
guarantees $\bar
\Psi(S)\ge 0$. We need to verify that all inequalities in
$D_{\F}(G,c)$ are maintained and that all tight constraints in
$D_{\F}(G,c)$ w.r.t. $\Psi$ are maintained. This trivially holds if
$uv\in E[V\setminus S]$. If $uv\in E[S]\setminus \supp(x)$, the
corresponding inequality is satisfied by both $\Pi$ and $\Psi$ and
hence also by their linear combinations.
If $uv\in E[S]\cap \supp(x)$, then $uv$ is tight for $\Psi$ by the optimality of
$\Psi$, and also for $\Pi$ by Lemma~\ref{lem:tight-inside}.

It remains to verify the constraint corresponding to edges $uv$ with $u\in S$, $v\in
V\setminus S$. The contribution of $\sum_{T\in \F:uv\in \delta(T), T\setminus S\neq
\emptyset}\Psi(T)$ is unchanged. The following claim completes the proof of
optimality.
\end{proof}

\begin{claim}\label{claim:bar-q-tightness}
$\bar \Psi_S(u)+\bar \Psi(S)\le \Psi_S(u)+\Psi(S)$ with equality whenever $uv\in
\supp(x)$.
\end{claim}
\begin{proof}
\begin{align*}
\bar \Psi(T)-\Psi(T)&=
\begin{cases}
& \lambda (\Pi(T)-\Psi(T))\mbox{ if $T\subsetneq S$},\\
& -\Delta\lambda\mbox{ if $T=S$}.
\end{cases}
\end{align*}
Thus,
\begin{align*}
\bar \Psi_S(u)+\bar \Psi(S)&=\lambda (\Pi_S(u)-\Psi_S(u))+\bar
\Psi(S)-\Psi(S)+\Psi_S(u)+\Psi(S)\\
&=\lambda(\Pi_S(u)-\Psi_S(u)-\Delta)+\Psi_S(u)+\Psi(S).
\end{align*}
Now, $\Pi_S(u)-\Psi_S(u)\le\Delta$, and   equality holds whenever
$uv\in\supp(x)\cap \delta(S)$ by the consistency of $\Psi$ and $\Pi$ inside $S$
(Lemma~\ref{lem:consistency-main}).
\end{proof}

\begin{corollary}\label{cor:combinatorial-positively-fitting-dual}
Let $\F$ be a critical family with $\Pi$ being an $\fc{\F}$ dual feasible solution.
Algorithm Positively-critical-dual-opt in Algorithm \ref{alg:positive-alg} transforms an
arbitrary dual optimal solution $\Psi$ to an $\pfc{\F}$ dual optimal solution in at
most $|\F|$ iterations.
\end{corollary}
\begin{proof}
The correctness of the algorithm follows by Lemma \ref{lem:bar-q}. We bound the running
time by showing that no set $S\in\F$ is processed twice. After a set $S$ is processed,
by Lemma \ref{lem:bar-q}, either $\Pi$ and $\Psi$ will be identical inside $S$ or
$\Psi(S)=0$. Once $\Pi$ and $\Psi$ become identical inside a set, it remains so during
all later iterations.

The value $\Psi(S)$ could be changed later only if we process a set $S'\supsetneq S$
after processing $S$. Let $S'$ be the first such set. 
At the iteration when $S$ was processed, by the maximal choice it follows that
$\Psi(S')=0$. Hence $\Psi(S')$ could become positive only if the algorithm had
processed a set $Z\supsetneq S'$, $Z\in\F$ between processing $S$ and $S'$, a
contradiction to the choice of $S'$.
\end{proof}

\subsection{Extremal dual solutions}\label{sec:extremal-dual-solutions}
In this section, we prove Lemma~\ref{lem:extremal-fit}. The end result of the iterative procedure of the previous section can also be achieved by optimizing over dual solutions. The key property of the objective function is that it puts less weight on larger laminar sets.

Assume $\F\subseteq\O$ is a
critical family, with $\Pi$ being an $\fc{\F}$ dual. Let $x$ be the unique optimal
solution to \ref{prog:P-F}.
Let ${\F}_x=\{S\in {\F}: x(\delta(S))=1\}$ the collection of tight sets for $x$. A
$\Pi$-extremal dual can be found by solving the following LP.
\begin{equation}
\tag{$D^*_{{\cal F}}$}\label{prog:D-star-F-i}
\begin{aligned}
\min h(\Psi,\Pi)=\sum_{S\in \V\cup{\cal F}_x}\frac{r(S)}{|S|}&\\
-r(S)\le \Psi(S)-\Pi(S)&\le r(S) \quad\forall S\in \V\cup \F_x\notag\\
\sum_{S\in\V\cup{\cal F}_x:uv\in \delta(S)}\Psi(S) &= c(uv) \quad \forall
uv\in \supp(x)\notag \\
\sum_{S\in\V\cup{\cal F}_x:uv\in \delta(S)}\Psi(S) &\le c(uv) \quad \forall
uv\in E\setminus \supp(x)\notag \\
\Psi(S)&\ge0\quad \forall S\in\F_x\notag
\end{aligned}
\end{equation}
The support of $\Psi$ is restricted to sets in $\V\cup\F_x$. Primal-dual slackness
implies that the feasible solutions to this program coincide with the optimal solutions
of \ref{prog:D-F}, hence an optimal solution to \ref{prog:D-star-F-i} is also an
optimal solution to \ref{prog:D-F}.

\begin{lemma}\label{lem:extreme}
Let $\F\subset \O$ be a critical family with $\Pi$ being an $\fc{\F}$ dual. Then, a
$\Pi$-extremal dual optimal solution is also an $\pfc{\F}$ dual optimal solution.
\end{lemma}
\begin{proof}
We will show that whenever $\Psi(S)>0$, the solutions $\Psi$ and $\Pi$ are identical
inside $S$.
Assume for a contradiction that this is not true for some $S\in \F$. Let
$\lambda=\min\{1,\Psi(S)/\Delta_{\Pi,\Psi}(S)\}$ if $\Delta_{\Pi,\Psi}(S)>0$ and
$\lambda=1$ if $\Delta_{\Pi,\Psi}(S)=0$. Define $\bar \Psi$ as in (\ref{eq:q-mod}). By
Lemma~\ref{lem:bar-q}, $\bar{\Psi}$ is also optimal to \ref{prog:D-F} and thus feasible
to \ref{prog:D-star-F-i}. We show $h(\bar \Psi,\Pi)<h(\Psi,\Pi)$, which is a
contradiction.

For every $T\in \V\cup \F_x$, let $\tau(T)=|\Psi(T)-\Pi(T)|-|\bar \Psi(T)-\Pi(T)|$.
With this notation,
\[
h(\Psi,\Pi)-h(\bar \Psi,\Pi)=\sum_{T\in\V\cup\F_x}\frac{\tau(T)}{|T|}.
\]
If $T\setminus S=\emptyset$, then $\bar \Psi(T)=\Psi(T)$ and thus $\tau(T)=0$. If
$T\subsetneq S$, $T\in \V\cup\F$, then $|\bar
\Psi(T)-\Pi(T)|=(1-\lambda)|\Psi(T)-\Pi(T)|$, and thus
$\tau(T)=\lambda|\Psi(T)-\Pi(T)|$. Since $\bar \Psi(S)=\Psi(S)-\Delta\lambda$, we have
$\tau(S)\ge -\Delta\lambda$.

Let us fix an arbitrary $u\in S$, and let $\gamma=\max_{T\subsetneq S:u\in T,
T\in\V\cup\F_x} |T|$.
\begin{align*}
h(\Psi,\Pi)-h(\bar \Psi,\Pi)
&=\sum_{T\in\V\cup\F_x}\frac{\tau(T)}{|T|}\\
&\ge \sum_{T\subsetneq S:u\in T,
T\in\V\cup\F_x}\frac{\tau(T)}{|T|}+\frac{\tau(S)}{|S|}\\
&\ge\frac{\lambda}{\gamma} \sum_{T\subsetneq S:u\in T,
T\in\V\cup\F_x}|\Psi(T)-\Pi(T)|-\frac{\Delta\lambda}{|S|}\\
&\ge
\frac{\lambda}{\gamma}\left(\Pi_S(u)-\Psi_S(u)\right)-\frac{\Delta\lambda}{|S|}.
\end{align*}

\noindent {\em Case 1:} If $\Delta>0$, then pick $u\in S$ satisfying
$\Pi_S(u)-\Psi_S(u)=\Delta$. Then the above inequalities give
\[
h(\Psi,\Pi)-h(\bar \Psi,\Pi)\ge  \Delta\lambda\left(\frac1\gamma-\frac1{|S|}\right)>0.
\]
The last inequality follows since $|S|>\gamma$.

\noindent {\em Case 2:} If $\Delta=0$, then $\lambda=1$ and therefore,
\[
h(\Psi,\Pi)-h(\bar \Psi,\Pi)\ge
\frac{1}{\gamma} \sum_{T\subsetneq S:u\in T,T\in\V\cup\F_x}|\Psi(T)-\Pi(T)|
\]
Now, if $\Pi$ and $\Psi$ are not identical inside $S$, then there exists a node $u\in
S$ for which the RHS is strictly positive.
Thus, in both cases, we get $h(\bar \Psi,\Pi)<h(\Psi,\Pi)$, a contradiction to the
optimality of $\Psi$ to \ref{prog:D-star-F-i}.
\end{proof}

\begin{proof}[Proof of Lemma~\ref{lem:extremal-fit}]
By Lemma~\ref{lem:pof-uniqueness-implies-half-integrality}, the unique optimal $x$ to
\ref{prog:P-F} is proper-half-integral. Lemma~\ref{lem:extreme} already shows that a
$\Gamma$-extremal dual solution $\Pi$ is also $\pfc{\F}$. We need to show that the next
family of cuts is a critical family. Recall that the set of cuts for the next round
is defined as $\H'\cup\H''$, where $\H'=\{T\in \F: \Pi(T)>0\}$, and $\H''$ is defined
based on some cycles in $\supp(x)$. We need to show that every set of $\H'\cup\H''$ is
$\Pi$-factor-critical. This is straightforward for sets of $\H'$ by the definition of
the $\pfc{\F}$ property.

It remains to show that the sets of $\H''$ are also  $\Pi$-factor-critical. These are
defined for odd cycles $C\in \supp(x)$. Now, $\hat C\in \H''$ is the union of $V(C)$
and the maximal sets $S_1,\ldots, S_\ell$ of $\H'$ intersecting $V(C)$. We have
$\Pi(S_j)>0$ for each $j=1,\ldots,\ell$ and hence $x(\delta(S_j))=1$.

Let $u\in \hat C$ be an arbitrary node; we will construct the $\Pi$-critical matching
$\hat M_u$ in $\hat C$. Let us contract all sets $S_1,\ldots,S_\ell$ to nodes $s_1,\ldots,s_\ell$
w.r.t. $\Pi$. We know by Lemma~\ref{lem:contract-factor-critical}(iii) that the image
$x'$ of $x$ is proper-half-integral and that the odd cycle $C$ projects to an odd cycle $C'$ in $\supp(x')$.
Further, notice that $\{s_1,\ldots,s_\ell\}\subseteq V(C')$, and therefore $V(C')$ is the image of the entire set $\hat C$.
Let $u'$ be the image of $u$; 
since $C'$ is an odd cycle, there is a perfect matching
$M'_{u'}\subseteq C'$ of the set $V(C')\setminus \{u'\}$.

Assume first $u\in S_j$ for some $1\le j\le \ell$. Then $u'=s_j$. The pre-image $\hat
M$ of $M'_{u'}$ in the original graph contains exactly one edge entering each $S_k$ for
$k\neq j$ and no edges entering $S_j$. Also, $\hat M\subseteq C$ and thus $\hat M$ consists of tight edges w.r.t. $\Pi$.
 Consider the $\Pi$-critical matching $M_{u}$ for
$u$ in $S_j$. For $k\neq j$, if $a_kb_k\in \hat M\cap \delta(S_k)$, $a_k\in S_k$, then,
let $M_{a_k}$ be the $\Pi$-critical matching for $a_k$ in $S_k$. The union of $\hat M$,
$M_u$ and the  $M_{a_k}$'s give a $\Pi$-critical matching for $u$ inside $\hat C$.

If $u\in \hat C \setminus(\cup_{j=1}^\ell S_j)$, then similarly there is a
$\Pi$-critical matching $M_{a_k}$ inside every $S_k$. The union of $\hat M$ and the
$M_{a_k}$'s give the $\Pi$-critical matching for $u$ inside $\hat C$. We also have
$\Pi(S)>0$ for all non-maximal sets $S\in\H'\cup\H''$ since the only sets with
$\Pi(S)=0$ are those in $\H''$, and they are all maximal ones.
\end{proof}

	
\section{Convergence}
The goal of this section is to prove Lemmas \ref{lem:odd-cycles} and
\ref{lem:strong-progress}. Lemma \ref{lem:odd-cycles} shows that the number of odd cycles in the support
is nonincreasing. Lemma \ref{lem:strong-progress} shows that in a sequence of iterations where the
number of cycles does not decrease, all the new cuts
added continue to be included in subsequent iterations (till the number of cycles
decreases). In order to establish Lemma \ref{lem:strong-progress}, it is sufficient to show that the extremal dual solution has non-zero values on cuts that were added after the last decrease in the number of odd cycles.

These structural properties are established as follows.
First we develop a primal-dual procedure that transforms a half-integral matching to satisfy a chosen subset of odd-set inequalities. 
Next, we apply this procedure starting with an appropriate primal/dual solution to obtain the optimal primal solution of the LP occurring in the cutting plane algorithm. The analysis of the procedure shows that the number of odd cycles in nonincreasing. The key ingredient in the proof of Lemma~\ref{lem:strong-progress} is showing 
that whenever the number of odd cycles remains the same, then the extremal dual solution occurring in the cutting plane algorithm must be the same as the dual solution found by this procedure. As a consequence, properties of the dual solution found by this procedure also carry over to the extremal dual solution found by the algorithm. 


\subsection{The half-integral matching procedure}\label{sec:edmonds-alg}
We use the terminology of Edmonds' weighted matching algorithm
\cite{Edmonds65} as described by Schrijver \cite[Vol A, Chapter 26]{Schrijver03}.
For a laminar family $\L\cup\K$, consider the following pair of primal and dual linear programs;
note that the primal differs from \ref{prog:P-F} by requiring that the degree of every set in $\K$ is precisely one, similar to the node constraints. 
\begin{multicols}{2}
\noindent
\begin{align}
\min& \sum_{uv\in E} c(uv) z(uv) \tag{$P_{\L}^{\K}(G,c)$}\label{prog:P-F-K}\\
z(\delta(u))&=1\quad\forall u\in V\notag\\
z(\delta(S))&=1\quad\forall S\in \K\notag\\
z(\delta(S))&\ge 1\quad \forall S\in \L\notag\\
z&\ge0\notag
\end{align}
\begin{align}
\max \sum_{S\in\V\cup{\cal L}\cup\K}\Lambda(S)& \tag{$D_{\L}^{\K}(G,c)$}\label{prog:D-F-K}\\
\sum_{S\in\V\cup{\cal L}\cup{\cal K}:uv\in \delta(S)}\Lambda(S) &\le c(uv) \quad \forall
uv\in E\notag \\
\Lambda(S)&\ge 0\quad\forall S\in\L \notag
\end{align}
\end{multicols}

We note that every feasible solution to \ref{prog:P-F-K} is also a feasible solution to $P_{\L\cup\K}(G,c)$, whereas a feasible solution to \ref{prog:D-F-K} is a feasible solution to $D_{\L\cup\K}(G,c)$ only if $\Lambda(S)\ge0$ for all sets $S\in\K$.

The aim of the procedure is that for a given laminar family $\F\cup\K$  satisfying certain structural properties, we wish to transform a pair of primal and dual feasible solutions to $(P_{\F}^{\emptyset}(G,c), D_{\F}^{\emptyset}(G,c))$ 
to optimal solutions to a pair of primal and dual optimal solutions to $(P_{\L}^{\K}(G,c), D_{\L}^{\K}(G,c))$ for some $\L\subseteq \F$. The following notion of a {\em valid configuration} encapsulates these structural properties. We say that $(\L,\K,z,\Lambda)$ form a valid configuration, if the following hold:

\begin{enumerate}[(A)]
\item $\L\cup\K\subset \O$ is a laminar family, and all sets in $\K$ are disjoint from each other and all sets in $\L$. $\Lambda$ is a feasible solution to \ref{prog:D-F-K} with $\Lambda(S)>0$ for all $S\in \L$. Further, every set $S\in \L\cup\K$ is $(G_{\Lambda},\L\cup\K)$-factor-critical, where $G_{\Lambda}$ denotes the graph of tight edges wrt $\Lambda$. 

\item $z$ is proper-half-integral, satisfying all constraints of \ref{prog:P-F-K} except that $z(\delta(S))=0$ may hold for some $S\in \K$. The support of $z$ is an odd cycle inside every such set $S$. Inside every other set $S\in \K\cup\L$, $\supp(z)$ spans all vertices in $S$ and is a disjoint union of edges and a (possibly empty) even path.

\item Every edge in $\supp(z)$ is tight for $\Lambda$, and $z(\delta(S))=1$ for every $S\in\L$.

\end{enumerate}

\begin{algorithm}[ht!]
\caption{Half-integral Matching Procedure}\label{alg:half-int}
\noindent {\em Input.} 
A graph $G$ with edge costs $c$, and a valid configuration
$(\F,\K,x,\Pi)$. 

\noindent {\em Output.} A valid configuration $(\L,\K,z,\Gamma)$ with $\L\subseteq \F$, and $z$ being a proper-half-integral optimal solution 
to \ref{prog:P-F-K}.

\begin{enumerate}

\item Initialize $z=x$, $\Lambda=\Pi$, $\L=\F$.
Let $G^*=({\cal V}^*,E^*)$, where $E^* \subseteq E$ are edges that are tight w.r.t. $\Lambda$,
and all maximal sets of $\L\cup \K$ w.r.t. $\Lambda$ are contracted;
$c^*$ and $z^*$ are defined by the contraction. Let $T\subseteq \V^*$
denote the set of exposed nodes in $z^*$, and 
let $R(\supseteq T)$ be the set of exposed nodes and nodes incident to $\frac 12$-edges in $z^*$.

\item While $T$ is not empty,\\
\noindent {\em Case I: There exists an alternating $T$-$R$-walk in $G^*$}. Let
$P=v_0\ldots
v_{2k+1}$ denote a shortest such walk.
\begin{enumerate}[(a)]
\item If $P$ is an alternating path, and $v_{2k+1}\in T$, then change
  $z$ by alternating along $P$.
\item If $P$ is an alternating path, and $v_{2k+1}\in R-T$, then
  let $C$ denote the odd cycle containing $v_{2k+1}$. Change $z$ by
  alternating along $P$, and replacing $z$ on $C$ by a blossom with
  base $v_{2k+1}$.
\item If $P$ is not a path, then by Claim~\ref{cl:walk-blossom}, it contains
   an even alternating path $P_1$ to a blossom $C$. Change $z$ by
   alternating along $P_1$, and setting $z^*(uv)=1/2$ on every edge of $C$.
\end{enumerate}

\noindent {\em Case II: There exists no alternating
  $T$-$R$-walk in $G^*$}. Define
\begin{align*}
{\cal B}^+&:=\{S\in \V^*: \exists \mbox{ an even alternating path from
$T$ to }S\},\\
{\cal B}^-&:=\{S\in \V^*: \exists \mbox{ an odd alternating path from
$T$ to }S\}.
\end{align*}
For some $\varepsilon>0$, reset
\begin{align*}
\Lambda(S):=
\begin{cases}
\Lambda(S)+\varepsilon&\mbox{ if }S\in{\cal B}^+,\\
\Lambda(S)-\varepsilon&\mbox{ if }S\in{\cal B}^-.
\end{cases}
\end{align*}
Choose $\varepsilon$ to be the maximum value such that $\Lambda$
remains feasible to \ref{prog:D-F-K}.
\begin{enumerate}[(a)]
\item If some new edge becomes tight, then $E^*$ is extended.
\item If $\Lambda(S)=0$ for some $S\in \L\cap {\cal B}^-$ after the
  modification, then unshrink the node $S$. Set $\L :=\L\setminus \{S\}$.
\end{enumerate}
\end{enumerate}
\end{algorithm}

The input to the procedure (see Algorithm ~\ref{alg:half-int}) will be a graph $G$ with costs $c$, and a valid configuration $(\F,\K,x,\Pi)$. The procedure is iterative. 
In each iteration, it maintains  a valid configuration $(\L,\K,z,\Lambda)$, where $\L\subseteq \F$; the set $\K$ never changes during the execution of the procedure.
We terminate once $z$ is feasible to \ref{prog:P-F-K}. The complementary slackness conditions (C) imply that if $z$ is feasible to \ref{prog:P-F-K}, then $(z,\Lambda)$ form an optimal primal-dual pair to (\ref{prog:P-F-K},\ref{prog:D-F-K}).

The procedure works on the graph $G^*=({\cal V}^*,E^*)$, obtained the following way from $G$: We
first remove every edge in $E$ that is not tight w.r.t. $\Lambda$, and then contract
all maximal sets of $\L\cup\K$ w.r.t. $\Lambda$. The node set of ${\cal V}^*$ is identified
with the pre-images. Let $c^*$ denote the contracted cost function and $z^*$ denote the image
of $z$. Since $E^*$ consists only of tight edges, $\Lambda(u)+\Lambda(v)=c^*(uv)$ for
every edge $uv\in E^*$. Let $T\subseteq {\cal V}^*$ denote the sets in
$\K$ for which $z(\delta(S))=0$; by property (B), 
 $T$ is the set of nodes in ${\cal V}^*$ that have degree 0 in $z^*$, whereas all other nodes have degree 1.

\begin{claim}\label{cl:z-star}
The vector $z^*$ defined above is proper-half-integral. Assuming the uniqueness condition (\ref{prop:uniqueness}), the number of odd cycles in $\supp(z^*)$ plus the number of exposed nodes in $z^*$ equals  $\o(z)$. 
\end{claim}
\begin{proof}
It is clear that $z^*$ is half-integral, and that the image of every
odd cycle in $\supp(z)$ is a cycle in $\supp(z^*)$ or an exposed node.  
The last part of property (B) implies that all these cycles in $z$ must be odd.
\end{proof}

In the execution of the procedure, we may decrease $\Lambda(S)$ to 0 for a set $S\in \L$. In this case, we remove $S$ from $\L$. We modify $G^*$, $c^*$ and $z^*$
accordingly. This operation will be referred as `unshrinking' $S$. New sets will never be added to $\L$, that is, no new sets will be shrunk after the initial contractions: $|{\cal V}^*|$ may only increase. In contrast, sets in $\K$ are never unshrunk and the family $\K$ does not change.

The procedure works by modifying the solution $z^*$ and the dual
solution $\Lambda^*$. An edge $uv\in E^*$ is called a
0-edge/$\frac{1}2$-edge/1-edge according to the value $z^*(uv)$. A
modification of $z^*$ in $G^*$ naturally extends to a modification of
$z$ in $G$. Indeed, if $S\in \Lambda$ is a shrunk node in $\V^*$, and $z^*$ is
modified so that there is an 1-edge incident to $S$ in $G^*$, then let
$u_1v_1$ be the pre-image of this edge in $G$, with $u_1\in S$. Then modify $z$ inside
$S$ to be identical with the $\Lambda$-critical-matching $M_{u_1}$ inside $S$. If there
are two half-edges incident to $S$ in $G^*$, then let $u_1v_1, u_2v_2$ be the pre-image
of these edges in $G$, with $u_1,u_2\in S$. Then modify $z$ inside $S$ to be identical
with the convex combination $(1/2)(M_{u_1}+M_{u_2})$ of the
$\Lambda$-critical-matchings $M_{u_1}$  and $M_{u_2}$ inside $S$. Note
that this modification preserves the second part of property (C).

A walk $P=v_0v_1v_2\ldots v_k$ in $G^*$ is called an alternating walk,
if every odd edge is a 0-edge and every even edge is a 1-edge. If
every node occurs in $P$ at most once, it is called an alternating
path. By {\em alternating along the path $P$}, we mean modifying
$z^*(v_iv_{i+1})$ to $1-z^*(v_iv_{i+1})$ on every edge of $P$. If $k$
is odd, $v_0=v_k$ and no other node occurs twice, then $P$ is called a
{\em blossom} with base $v_0$. The following claim is
straightforward.

\begin{claim}[{\cite[Thm 24.3]{Schrijver03}}]\label{cl:walk-blossom}
Let $P=v_0v_1\ldots v_{2k+1}$ be an alternating walk. Either $P$ is an
alternating path, or it contains a blossom $C$ and an even alternating
path from $v_0$ to the base of the blossom.
\proofbox
\end{claim}

The procedure is described in Algorithm \ref{alg:half-int}.
Let us note that in \cite{our-matching-alg} we  extend it to a ``complete'' algorithm to find 
a minimum-cost perfect matching where the intermediate solutions are half-integral and satisfy the degree constraints for all vertices.

\begin{figure}[ht]
\centering
\includegraphics[scale=0.4]{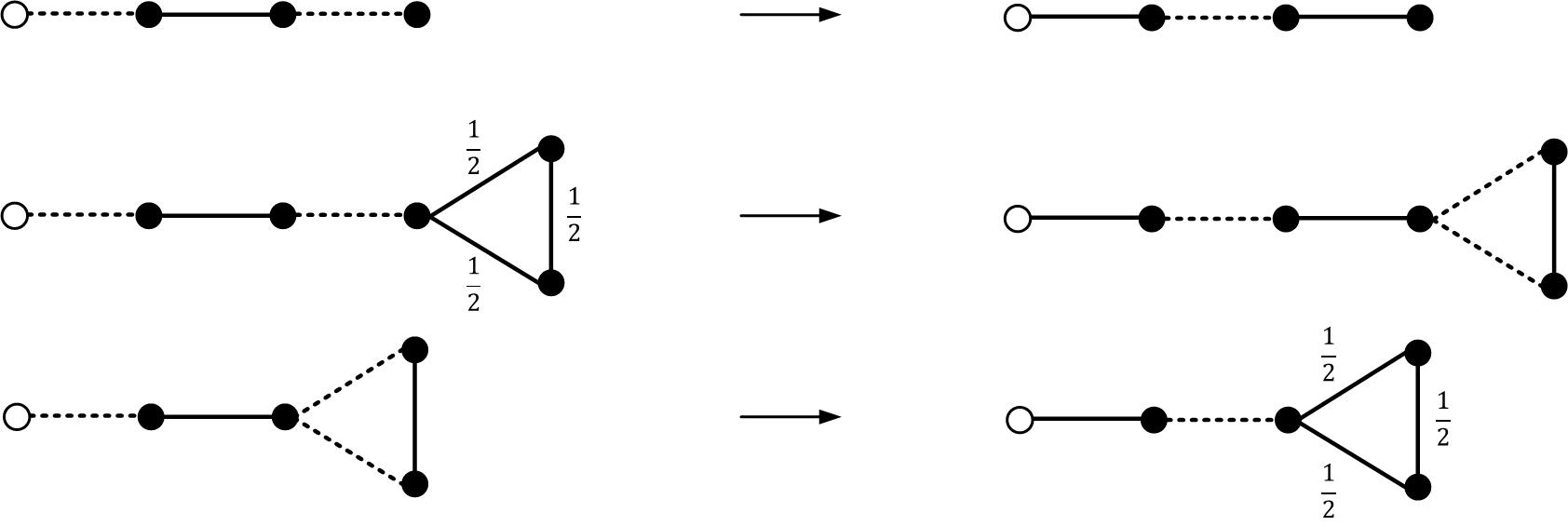}
\caption{The possible modifications in the Half-integral Matching Procedure.}
\label{fig:half-integral-primal-dual}
\end{figure}

The scenarios in {\em Case I} are illustrated in Figure~\ref{fig:half-integral-primal-dual}. In Case II, we observe that $T\subseteq {\cal B}^{+}$ and further, ${\cal B}^{+}\cap {\cal B}^-=\emptyset$ (otherwise, there existed a $T-T$ alternating walk and hence we have case I). The following claim is easy to verify. Note that $\L\cup\K$ will always be a critical family because of Claim~\ref{claim:downwards}. 

\begin{claim}
In every iteration of the procedure, $(\L,\K,z,\Lambda)$ is a valid configuration.
\end{claim}


The key to the proof of Lemma~\ref{lem:odd-cycles} is the following lemma, showing that $\o(z)$ is non-increasing during the execution of the procedure.

\begin{lemma}\label{lem:edmonds-decrease}
Let $z$ be the solution at the beginning of execution of an arbitrary iteration of the procedure, and let $\alpha$ be the number
of odd cycles  in $\supp(x)$ disjoint from all members of $\K$, that are absent in $\supp(z)$. Then
$\o(x)\ge \o(z)+2\alpha$. Further, if $\o(x)=\o(z)$, then cases I(a) and I(b) are never executed.
\end{lemma}
\begin{proof}
We will investigate how  the number of odd cycles in
$\supp(z^*)$ plus the number of exposed nodes change; by Claim~\ref{cl:z-star}, this equals $\o(z)$.
In Case I(a), the number of exposed nodes decreases by two. In Case I(b), both the number
of exposed nodes and the number of cycles decrease by one. In Case I(c), the number of
exposed nodes decreases by one, but we obtain a new odd cycle, hence
the total quantity remains
unchanged. In Case II, the primal solution is not modified at all.

Further, the cycles in $\supp(z^*)$ are in one-to-one correspondence with the cycles in $\supp(z)$ that are not contained inside some member of $\K$. Such a cycle can be removed only by performing the operation in Case I(b). This must be executed $\alpha$ times, therefore $\o(z)\le\o(x)-2\alpha$.
\end{proof}

Lemma~\ref{lem:odd-cycles} will be an immediate consequence of the next lemma.
For the proof of this next lemma, we assume that the procedure terminates in finite number of iterations. In the next section,
we will show that the procedure indeed terminates in strongly polynomial time.

\begin{lemma}\label{lem:inductive-odd}
Assume (\ref{prop:uniqueness}) holds. Let $\F$ be a critical family, and let $x$ be an optimal solution to
\ref{prog:P-F}, and $\Pi$ an $\pfc{\F}$ dual optimal solution to
\ref{prog:D-F}. Define the sets $\H'$ and $\H''$ as in steps 2(b) and (c) in Algorithm C-P matching (Algorithm~\ref{alg:main-alg}),
and let $\H=\H'\cup\H''$. Let $y$ be an optimal solution to $P_{\H}(G,c)$ and let $\Psi$  be an $\pfc{\H}$  optimal solution to $P_{\H}(G,c)$.
Then $\o(y)\le \o(x)$, and if $\o(y)=\o(x)$ then $\Psi(S)>0$ for every $S\in \H''$.
\end{lemma}

\begin{proof}

We first note that Lemma~\ref{lem:extremal-fit} guarantees that $\H$ is a critical family; hence the existence of $\Psi$ is guaranteed.
Further, (\ref{prop:uniqueness}) guarantees the uniqueness of $x$ and $y$.
To prove the lemma by contradiction, consider a counterexample $(G,c,\F)$ with $|V|$ minimal. That is, either $\o(y)>\o(x)$, or $\o(y)=\o(x)$ but $\Psi(S)=0$ for some $S\in \H''$. 
\begin{claim}\label{cl:psi-zero}
$\Psi(S)=0$ for every $S\in \H'$. 
\end{claim}
\begin{proof}
Consider a set $S\in \H'$ with $\Psi(S)>0$. 
We first observe that $\Pi$ and $\Psi$ are feasible solutions to $D_{\H}(G,c)$. Further, if $S\in \H'$, then $S$ is both $(\Pi,\H)$-factor-critical and $(\Psi,\H)$-factor-critical. Hence, by Lemma~\ref{lem:factorcrit-ident}, $\Pi$ and $\Psi$ are identical inside $S$. 

Let us contract the set $S$ with respect to $\Pi$; this is equivalent to contracting with respect to $\Psi$ since $\Pi$ and $\Psi$ are identical inside $S$. Let $\hat G$, $\hat c$, $\hat \F$, $\hat x$, $\hat \Pi$, $\hat \H$, $\hat y$, $\hat \Psi$ denote the respective images to the contracted instance.
Lemma~\ref{lem:contract-factor-critical} guarantees that in the contracted instances, $\hat\Pi$ and $\hat\Psi$ are critical families, $\hat x$ and $\hat y$ are optimal solutions to $P_{\hat \F}(\hat G,\hat c)$ and  $P_{\hat \H}(\hat G,\hat c)$, respectively, and $\o(\hat x)=\o(x)$, $\o(\hat y)=\o(y)$. Furthermore, $\hat\Pi$ is an $\pfc{\hat \F}$ dual optimal solution to $D_{\hat \F}(\hat G,\hat c)$. Now, by taking $\hat \H=\hat \H'\cup \hat \H''$, where the sets $\hat \H'$ and $\hat \H''$ are the ones obtained in steps 2(b) and (c) in Algorithm C-P matching applied for $(\hat G,\hat c, \hat \F)$, we obtain a smaller counterexample. That is, if the counterexample $(G,c,\F)$ is such that $\o(y)>\o(x)$, then $\o(\hat{y})>\o(\hat{x})$ in the contracted instance; if the counterexample $(G,c,\F)$ is such that $\o(y)=\o(x)$ but $\Psi(S)=0$ for some $S\in \H''$, then $\o(\hat{y})=\o(\hat{x})$, $\hat{\Psi}(\hat{S})=0$ and $\hat{S}\in \hat{\H}''$.

\end{proof}

Let us define $\K:=\{S\in \H'': \Psi(S)>0\}$, and  apply the Half-integral Matching Procedure with input $(\H',\K,x,\Pi)$. 
\begin{claim}\label{cl:valid-config}
$(\H',\K,x,\Pi)$ is a valid configuration.
\end{claim}
\begin{proof}
We verify only the nontrivial properties. 
We show that all sets in $\H'$ are disjoint from all sets in $\K$.
Consider a set $S\in \K$; thus $\Psi(S)>0$. Hence $S$ is both $(\Pi,\H)$-factor-critical and $(\Psi,\H)$-factor-critical, and therefore Lemma~\ref{lem:factorcrit-ident} is applicable. Consequently $\Psi$ and $\Pi$ must be identical inside $S$. For the sake of contradiction, assume there exists a set $T\subsetneq S$, $T\in \H'$. Then $\Psi(T)=\Pi(T)>0$, contradicting Claim \ref{cl:psi-zero}.

Thus, the sets in $\K$ are disjoint from each other as well as the ones in $\H'$. This immediately implies that $\supp(x)$ is an odd cycle inside every set in $\K$ since $\K\subseteq \H''$. Finally, by Claim \ref{lem:contract-factor-critical} (iii), we have that $\supp(x)$ inside each set $T$ is a disjoint union of edges and an even path. 
\end{proof}
Let $(\L,\K,z,\Lambda)$ denote the output of the Half-integral Matching Procedure applied to $(\H',\K,x,\Pi)$, where $\L\subseteq \H'$ and $z$ is an optimal solution to \ref{prog:P-F-K}.
Lemma~\ref{lem:edmonds-decrease} implies that $\o(z)\le \o(x)$. 

\[
(\H',\K,x,\Pi) \xrightarrow[\quad\substack{\mbox{Half-integral matching}\\ \mbox{procedure}}\quad]{} (\L,\K,z,\Lambda)
\]

\begin{claim} $y=z$. \end{claim}
\begin{proof}
We know that $y$ is an optimal solution to $P_{\H}(G,c)$. We observe that $y$ is also an optimal solution to $P_{\K}(G,c)$. This follows since by Claim \ref{cl:psi-zero}, we have $S\in \K$ whenever $\Psi(S)>0$.
 Hence, the solution $\Psi$ is feasible to $D_{\K}(G,c)$ and $(y,\Psi)$ satisfy complementary slackness for primal-dual pair of programs $(P_{\K}(G,c),D_{\K}(G,c))$.


Further, $y$ is a feasible solution to $P_{\L}^\K(G,c)$ by the choice of $\K$. We now note that the feasible region of $P_{\L}^\K(G,c)$ is a subset of the feasible region of $P_{\L\cup \K}(G,c)$, which in turn, is a subset of the feasible region of $P_{\K}(G,c)$.
Therefore, $y$ is also optimal to all three programs, and
consequently, $c(y)=c(z)$. This in turn implies that $z$ is also
optimal to $P_{\L\cup \K}(G,c)$. Since $\L\cup\K$ is a critical
family, the uniqueness assumption (\ref{prop:uniqueness}) implies that
$y$ and $z$ must be identical.
\end{proof}

Hence we have $\o(y)\le \o(x)$. Since $(G,c,\F)$ was a counterexample, we must have $\o(y)=\o(x)$ and $\Psi(S)=0$ for some $S\in \H''$. This is equivalent to $\K$ being a strict subset of $\H''$, that is, there exists a cycle $C$ in $\supp(x)$ that is disjoint from all sets in $\K$. 
Let $\hat C=V(C)\cup \left(\cup_{T\in \H':T\cap V(C)\neq \emptyset} T\right)$ denote the set included in $\H''$ corresponding to the cycle $C$ in Algorithm C-P matching. 
The following claim implies that $y(\delta(\hat{C}))=0$ contradicting the feasibility of $y$ for $P_{\H}(G,c)$ and thus completes the proof. 

\begin{claim}\label{cor:non-decreasing-not-executed-cases}
If $\o(x)=\o(z)$, then $z(\delta(\hat C))=0$.
\end{claim}
\begin{proof}
By Lemma \ref{lem:edmonds-decrease}, $C\in \supp(z)$ for otherwise, $\o(z)\le \o(x)-2$. 
Let $C^*$ denote the cycle corresponding to $C$ in the contracted graph at the beginning of the first iteration. 
Consider the case where all sets $T\in \H'$ that intersect $V(C)$ also belong to $\L$. Then all these sets are contracted in the current $G^*$. This implies that $C^*$ is present in $\supp(z^*)$ and hence $z^*(\delta(V(C^*)))=0$. Therefore $z(\delta(\hat C))=0$. 

Consider the earliest iteration such that a set $T\in \H'$ intersecting $V(C)$ leaves $\L$. This means that the dual value on $T$ decreased to zero in this iteration, that is $T^*\in \mathcal{B}^-$. Now, this is possible only if the cycle $C^*$ is absent at the beginning of this iteration. A cycle in the support of the primal solution could have been broken only by the execution of Case I(b). This contradicts Lemma \ref{lem:edmonds-decrease}.
\end{proof}

\end{proof}

\subsection{The non-decreasing scenario}\label{sec:non-decreasing}
In this section we analyze further properties of the Half-integral
Matching procedure to show termination, and also to establish 
structural properties needed for the proof of
Lemma~\ref{lem:strong-progress} in the subsequent section.
If $\o(z)$ is unchanged during a certain number of iterations of
 the procedure, we say that these iterations form a {\em non-decreasing phase}. We say that the procedure
itself is non-decreasing, if $\o(z)$ does not decrease anytime. We show that
every non-decreasing phase may contain at most $|V|+|\F|$ iterations
and therefore the procedure terminates in strongly polynomial time.

Let us now analyze the first non-decreasing phase ${\cal P}$ of the procedure,
starting from the input $x$. These results will also be valid for later
non-decreasing phases.
Consider an intermediate iteration with a valid configuration $(\L,\K,z,\Lambda)$ at the beginning of the iteration. Let $T\subseteq\V^*$ denote the set of
exposed nodes in $z^*$, and $R\supseteq T$ is the set of exposed nodes and the node sets of
the $1/2$-cycles.
Let us define the set of outer/inner nodes of $G^*$ as those
having even/odd length alternating walk from $R$ in $G^*$. Let ${\cal
  N}_o$ and ${\cal N}_i$ denote
their sets, respectively. Clearly, ${\cal B}^+\subseteq {\cal N}_o$,
${\cal B}^-\subseteq
{\cal N}_i$ in Case II of the algorithm.

\begin{lemma}\label{lem:inner-outer}
If $\cal P$  is a non-decreasing phase, then if a node in $\V^*$ is
outer  in any
iteration of phase ${\cal P}$, it remains a node in $\V^*$ and an outer node
in every later iteration of ${\cal P}$.
If a node is inner in any iteration of ${\cal P}$, then in any later iteration
of ${\cal P}$, it is either an inner node, or it has been unshrunk
in an intermediate iteration.
\end{lemma}
\begin{proof}
Since ${\cal P}$ is a non-decreasing phase, Cases I(a) and (b) can never be
executed. We show that the claimed properties are maintained during an iteration.

In Case I(c), a new odd cycle $C$ is created, and thus the vertex set of $C$ is added to
$R$. Let $P_1=v_0\ldots v_{2\ell}$ denote the even alternating path with $v_0\in
T$, $v_{2\ell}\in C$. If a node $u\in \V^*$ had an even/odd alternating walk from
$v_0$ before changing the solution, it will have an even/odd walk alternating from
$v_{2\ell}\in R$ after changing the solution.

In Case II, the alternating paths from $T$
to the nodes in ${\cal B}^-$ and ${\cal B}^+$ are maintained when the duals are
changed. The only nontrivial case is when a set $S$ is
unshrunk; then all inner and outer nodes maintain their inner and outer property for the
following reason:
if $u_1v_1$ is a 1-edge and $u_2v_2$ is a 0-edge entering $S$ after
unshrinking, with $u_1,u_2\in S$, we claim that there exists an even
alternating path inside $S$ from $u_1$ to $u_2$ using only tight edges
wrt $\Lambda$.
Indeed, during the unshrinking, we modify $z$ to $M_{u_1}$ inside $S$.
Also, by the $\Lambda$-factor-critical property, all edges of
$M_{u_2}$ are tight w.r.t. $\Lambda$. Hence the symmetric difference
of $M_{u_1}$ and $M_{u_2}$ contains an even alternating path from $u_1$ to
$u_2$.

We have to check that vertices in ${\cal N}_o-{\cal B}^+$ and ${\cal N}_i-{\cal
  B}^-$ also maintain their outer and inner property. These are the nodes having
  even/odd
alternating paths from an odd cycle, but not from exposed
nodes. The nodes in these paths are disjoint from ${\cal B}^-\cup
{\cal B}^+$ and are thus
maintained. Indeed, if $({\cal B}^-\cap {\cal N}_o)\setminus{\cal
  B}^+\neq\emptyset$ or $({\cal B}^+\cap {\cal N}_i)\setminus{\cal
  B}^-\neq\emptyset$,  then we would get an alternating walk from $T$ to an odd cycle,
  giving the forbidden Case I(b).
\end{proof}

The termination of the algorithm is guaranteed by the following simple corollary.
\begin{corollary}\label{cor:half-integral-primal-dual-termination}
The non-decreasing phase $\cal P$ may consist of at most $|V|+|\F|$ iterations.
\end{corollary}
\begin{proof}
Case I may occur at most $|\K|$ times as it decreases the number of
exposed nodes. In Case II, either ${\cal N}_i$ is extended, or a set
is unshrunk. By Lemma \ref{lem:inner-outer}, the first scenario may
occur at most $|V|-|\K|$ times and the
second at most $|\F|$ times.
\end{proof}

In the rest of the section, we focus on the case when the entire
procedure is non-decreasing.
\begin{lemma}\label{lem:edmonds-equal}
Assume the Half-integral Matching procedure is non-decreasing with input $(\F,\K,x,\Pi)$ and output
$(\L,\K,z,\Lambda)$.
Let ${\cal N}_o$ and ${\cal N}_i$ denote the final sets of outer and inner nodes in
$G^*$.
\begin{itemize}
\item If $\Lambda(S)>\Pi(S)$ then $S$ is an outer node in $\V^*$.
\item If $\Lambda(S)<\Pi(S)$, then either $S\in \F\setminus\L$,  (that is,
  $S$ was unshrunk during the algorithm and $\Lambda(S)=0$)
or
$S$ is an inner node in $\V^*$, or $S$ is a node in
  $\V^*$ incident to an odd cycle in $\supp(z)$.
\end{itemize}
\end{lemma}
\begin{proof}
If $\Lambda(S)>\Pi(S)$, then $S\in {\cal B}^+$ in some
iteration of the algorithm. By Lemma~\ref{lem:inner-outer}, this
remains an outer node in all later iterations. The conclusion follows
similarly for $\Lambda(S)<\Pi(S)$.
\end{proof}

\begin{lemma}\label{lem:inner-bigger}
Assume the Half-integral Matching procedure is non-decreasing. Let $(\L,\K,z,\Lambda)$ be 
the terminating valid configuration, and let 
$G^*$ be the
corresponding contracted graph, ${\cal N}_o$ and ${\cal N}_i$ be the sets of outer
and inner nodes. Let $\Theta:\V^*\rightarrow\R$ be an arbitrary optimal
solution to the dual $D_0(G^*,c^*)$ of the bipartite relaxation. If $S\in \V^*$ is
incident to an odd cycle in $\supp(z)$, then $\Lambda(S)=\Theta(S)$.
Further $S\in {\cal N}_o$ implies
$\Lambda(S)\le \Theta(S)$, and $S\in {\cal N}_i$ implies $\Lambda(S)\ge \Theta(S)$.
\end{lemma}
\begin{proof}
For $S\in {\cal N}_o\cup {\cal N}_i$, let $\ell(S)$ be the length of the shortest
alternating path. The proof is by induction on $\ell(S)$.
Recall that there are no exposed nodes in $z$, hence $\ell(S)=0$ means that $S$ is
contained in an odd cycle $C$. Then $\Theta(S)=\Lambda(S)$ is a consequence of
Lemma~\ref{lem:factorcrit-ident}: both $\Theta$
and $\Lambda$ are optimal dual solutions in $G^*$, and an odd cycle in the
support of the primal optimum $z$ is both $\Lambda$-factor-critical and
$\Theta$-factor-critical.

For the induction step, assume the claim for
$\ell(S)\le i$. Consider a node $U\in\V^*$ with $\ell(U)=i+1$.
There must be an edge $f$ in $E^*$ between $S$ and $U$ for some $S$ with $\ell(S)=i$.
This is a 0-edge if $i$ is even and a 1-edge if $i$ is odd.

Assume first $i$ is even.
By induction, $\Lambda(S)\le \Theta(S)$. The edge $f$ is tight for $\Lambda$, and
$\Theta(S)+\Theta(U)\le c^*(f)$. Consequently, $\Lambda(U)\ge \Theta(U)$ follows.
Next, assume $i$ is odd. Then $\Lambda(S)\ge \Theta(S)$ by induction. Then,
$\Lambda(U)\le \Theta(U)$ follows as $f$ is tight for both $\Lambda$ and $\Theta$.
\end{proof}

\newcommand{\out}{\mbox{Outer}}

\subsection{Proof of Lemma~\ref{lem:strong-progress}}
We follow a similar strategy to prove Lemma~\ref{lem:strong-progress} as for Lemma~\ref{lem:odd-cycles}: we formulate an analogue of 
Lemma~\ref{lem:inductive-odd}. For that, we need the notion of outer nodes with respect to intermediate solutions in Algorithm C-P-matching.
Let $\F$ be a critical family, and let $x$ be an optimal solution to
\ref{prog:P-F}, and $\Pi$ an $\pfc{\F}$ dual optimal solution to
\ref{prog:D-F}. Let us define the sets $\H'$ and $\H''$ as in steps 2(b) and (c) in Algorithm C-P-matching (Algorithm~\ref{alg:main-alg}),
and let $\H=\H'\cup\H''$. 

Let us consider the graph $G^*=(\V^*,E^*)$ obtained by first deleting all edges that are non-tight wrt $\Pi$ and contracting all maximal sets of $\H'\cup \H''$ wrt $\Pi$. We observe that the resulting instance is identical to the contracted graph obtained in the execution of the Half-integral Matching procedure using the configuration $(\H',\H'',x,\Pi)$. Even though the configuration may not be a valid one, the contraction operation and the outer property of the nodes in $G^*$ are well-defined. 
Let $\out(x,\Pi,\H)$ denote the sets in $\H$ contained in the outer nodes of $G^*$, that is,
\[
\out(x,\Pi,\H):=\{S\in \H: \exists T\in \V^*, T\mbox{ is an outer node in }G^*, S\subseteq T\mbox{ and }\Pi(S)>0\}.
\]
Observe that $\H''\subseteq \out(x,\Pi,\H)$ since every set in $\H''$ is exposed in $G^*$. 

\begin{lemma}\label{lem:inductive-outer}
Assume (\ref{prop:uniqueness}) holds, and use the notation $x,\Pi,\F,\H$ as above. Let $y$ be an optimal solution to $P_{\H}(G,c)$ and let $\Psi$ be a $\Pi$-extremal optimal solution to $D_{\H}(G,c)$. Let $\H_{\text{next}}$ denote the next round of cuts to be imposed by Algorithm C-P-matching from $(y,\Psi)$. If $\o(y)=\o(x)$ then $\out(x,\Pi,\H)\subseteq \out(y,\Psi,\H_{\text{next}})$.
\end{lemma}

\begin{proof}[Proof of Lemma~\ref{lem:strong-progress}]
Using the notation of the lemma, let $\F_k$ be the set of constraints in iteration $k$, and $\H''_k=\F_{k+1}\setminus\F_k$. Note that 
by definition $\H''_k\subseteq \out(x_k,\Pi_k, \F_{k+1})$. Also, by Lemma~\ref{lem:inductive-outer}, if a sets enters $\out(x_k,\Pi_k, \F_{k+1})$ then it remains in $\out(x_\ell,\Pi_\ell, \F_{\ell+1})$ for all $\ell>k$ if all iterations between the $k$'th and the $\ell$'th are non-decreasing. Since $\Pi_\ell(S)>0$ for all 
sets $S\in \out(x_\ell,\Pi_\ell, \F_{\ell+1})$, all cuts corresponding to sets in $\out(x_\ell,\Pi_\ell, \F_{\ell+1})$ will be imposed in the $\ell$'th iteration, as required.
\end{proof}

The rest of the section is dedicated to proving Lemma~\ref{lem:inductive-outer}. For a contradiction, let us choose a counterexample $(G,c,\F)$ with $|V|$ minimal.
By Lemma~\ref{lem:inductive-odd}, $\Psi(S)>0$ for every $S\in \H''$. By taking $\K=\H''$, Claims \ref{cl:psi-zero} and \ref{cl:valid-config}, and their proofs hold verbatim. Hence the support of $\Psi$ is identical to $\H''$. 
Hence, $(\H',\H'',x,\Pi)$ is a valid configuration. Let us run the Half-Integral Matching Procedure with input  $(\H',\H'',x,\Pi)$, terminating with $(\L,\H'',z,\Lambda)$.
As in the proof of Lemma~\ref{lem:inductive-odd}, we have $y=z$. The key part of the proof is showing that $\Lambda$ also coincides with $\Psi$.

\begin{lemma}\label{lem:dual-identity}
We have $\Lambda=\Psi$.
\end{lemma}

Lemma~\ref{lem:inductive-outer} is a straightforward consequence of Lemmas \ref{lem:dual-identity} and \ref{lem:inner-outer}.

\begin{proof}[Proof of Lemma \ref{lem:inductive-outer}]
By Lemma \ref{lem:dual-identity}, $\L=\emptyset$ since the support of $\Psi$ is $\H''$. 
By Lemma \ref{lem:inner-outer}, the sets $S\in $ Outer$(x,\Pi,\H)$ are outer nodes in the graph $G^*$ obtained after deleting all edges that are non-tight wrt $\Psi$ and contracting all maximal sets of $\H''$ wrt $\Psi$. Also, $\Psi(S)=\Lambda(S)\ge \Pi(S)>0$ for every $S\in$ Outer$(x,\Pi,\H)$ since such sets are outer nodes at the beginning of the execution of the Half-Integral Matching Procedure and the dual values on outer nodes are non-decreasing throughout the execution of the procedure. Therefore, Outer$(x,\Pi,\H)\subseteq$ Outer$(y,\Psi,\H'')$. Further, by the choice of cuts, $\H''\subseteq \H_{\text{next}}$ since $\H''$ is the support of $\Psi$. 
Hence, Outer$(x,\Pi,\H)\subseteq$ Outer$(y,\Psi,\H_{\text{next}})$.

\end{proof}


As a first step in the proof of Lemma \ref{lem:dual-identity}, we show that 
 $\Lambda$ is optimal to $D_{\H}(G,c)$;
note that $\Psi$ was defined as the extremal dual solution to this problem.
\begin{claim}\label{claim:lambda-0-dual-opt}
$\Lambda$ is an optimal solution to $D_{\H}(G,c)$.
\end{claim}
\begin{proof}
It is sufficient to show that $\Lambda$ is feasible to $D_{\H}(G,c)$. Optimality then follows by the complementary slackness property (C) between
$z=y$ and $\Lambda$.
We have $\Lambda(S)=0$ for all $S\in \H\setminus \L$, and
$\Lambda(S)\ge0$ for all $S\in\L$. Hence we only need to prove that
$\Lambda(S)\ge0$ for all $S\in\H''$. This follows by Lemma~\ref{lem:inner-outer} since every such $S$ is an outer node at the beginning of the algorithm.
\end{proof}

Let us study the image of $\Psi$ in the contracted graph
$(G^*,c^*)$, where this graph is defined with respect to the terminating solution $(\L,\H'',z,\Lambda)$. Let us define
$\Theta:\V^*\rightarrow \R$ as follows:
\begin{align*}
\Theta(S)&=
\begin{cases}
 \Psi(S)-\Delta_{\Lambda,\Psi}(S) &\mbox{ if }S\in \L\cap\V^*,\\
 \Psi(S) &\mbox{ if }S\in \V^*\setminus  \L.
\end{cases}
\end{align*}
Note that $\L\cap\V^*$ is the set of maximal sets in $\L$.
\begin{claim}\label{claim:theta}
$\Theta$ is an optimal solution to
$D_0(G^*,c^*)$. Further,
$\Theta(S)\le 0$ holds for every $S\in \L\cap\V^*$.
\end{claim}
\begin{proof}
$G^*$ arises from $G$ by contracting the maximal sets in $\L\cup\H''$ w.r.t. $\Lambda$.
We note that $\Psi$ is a $\pfc{\L\cup\H''}$ dual optimal solution to
$D_{\L\cup\H''}(G,c)$.
Optimality is because $\Psi$ is an optimal dual solution to $D_{\L}^{\H''}(G,c)$,
since it is feasible to this program, and it is optimal to
$D_{\H}(G,c)$ that has a larger feasible region; the
$\pfc{\L\cup\H''}$ property follows by Lemma~\ref{lem:extremal-fit}.

First consider a set $S\in\H''$.
Inside $S$, Lemma~\ref{lem:factorcrit-ident} implies that
$\Pi$ and $\Psi$ are identical. Further, $\Pi$ and $\Lambda$ are
also identical inside such sets, because the sets
in $\H''$ are never unshrunk during the procedure and hence the dual
values inside do not change. 
Hence contracting such a set w.r.t. $\Lambda$ is the same as
contracting it w.r.t. $\Psi$.

Consider now a maximal set $S\in \L$. By property (A), $S$ is
$(G_\Lambda,\L\cup\H'')$-factor-critical and  $\Lambda(T)>0$ for every $T\subseteq S, T\in \L$.
 By property (C), we have $z(\delta(S))=1$. 
By Claim~\ref{claim:lambda-0-dual-opt}, $\Lambda$ is an optimal (in particular, feasible) solution to $D_{\H}(G,c)$.
Thus Lemma \ref{lem:consistency-main}
is applicable and shows that $\Lambda$ and $\Psi$ are consistent
inside $S$ with $\Delta_{\Lambda,\Psi}(S)\ge 0$.
By Lemma~\ref{lem:contract-consistent},  we may contract such a set $S$ w.r.t. $\Lambda$, and we
obtain an optimal dual solution from $\Psi$ by subtracting
$\Delta_{\Lambda,\Psi}(S)$ from $\Psi(S)$ and leaving the values on
all other sets unchanged.

Applying the above arguments one-by-one for all maximal members of
$\L\cup \H''$ we can conclude that  $\Theta$ is an optimum dual solution to
$D_0(G^*,c^*)$. The second part follows since  $\Psi(S)=0$ and
$\Delta_{\Lambda,\Psi}(S)\ge 0$ if $S\in\L$.
\end{proof}

\begin{claim}\label{claim:psi-smaller}
For every $S\in  \V\cup \H$,
$|\Lambda(S)-\Pi(S)|\le |\Psi(S)-\Pi(S)|$ and equality holds only if
$\Lambda(S)=\Psi(S)$.
\end{claim}
\begin{proof}
The claim will follow by showing that for every $S\in \V\cup \H$, either
$\Pi(S)\le \Lambda(S)\le\Psi(S)$ or $\Pi(S)\ge
\Lambda(S)\ge\Psi(S)$.

First, if $\Lambda(S)>\Pi(S)$, then by
Lemma~\ref{lem:edmonds-equal}, we
have that $S\in \V^*$ and $S\in {\cal N}_o$. Consequently, by
Lemma~\ref{lem:inner-bigger}, $\Theta(S)\ge \Lambda(S)$.  If $S\in \L$, then
$0\ge\Theta(S)\ge \Lambda(S)$ using Claim~\ref{claim:theta}. 
Otherwise,
$\Psi(S)=\Theta(S)\ge \Lambda(S)>\Pi(S)$.

If $\Lambda(S)<\Pi(S)$, then by Lemma~\ref{lem:edmonds-equal}, we have
that either {\em(1)} $S\in \H'\setminus\L$, that is,
$\Lambda(S)=0$ and $S$ was unshrunk or
 {\em(2)} $S\in {\cal N}_i$ or {\em(3)} $S\in \V^*$ and $S$ is incident to an odd
cycle $C$ in $\supp(z)$. 
Note that $S\notin\H''$ since all sets in $\H''$ are in ${\cal N}_o$ and $\Pi(S)>\Lambda(S)$. 
In case {\em(1)}, we have $\Psi(S)=0=\Lambda(S)<\Pi(S)$.
In both cases {\em(2)} and {\em(3)}, Lemma~\ref{lem:inner-bigger} gives $\Theta(S)\le
\Lambda(S)$. If $S\in  \L$, then $\Psi(S)=0\le \Lambda(S)<\Pi(S)$.
If $S\in \V\cup(\H'\setminus \L)$, then $\Psi(S)=\Theta(S)\le
\Lambda(S)\le \Pi(S)$.
\end{proof}

\begin{proof}[Proof of Lemma~\ref{lem:dual-identity}]
The proof is straightforward by the above two claims: By
Claim~\ref{claim:lambda-0-dual-opt}, $\Lambda$ is optimal to
$D_{\H}(G,c)$, and Claim~\ref{claim:psi-smaller}
shows that $h(\Lambda,\Pi)\le h(\Psi,\Pi)$, and equality can hold only if
$\Lambda$ and $\Psi$ are identical.
\end{proof}

As remarked above, this completes the proof of Lemma~\ref{lem:inductive-outer}, and hence of Lemma~\ref{lem:strong-progress}.


\section{Uniqueness}\label{sec:uniqueness}
In Section~\ref{sec:pos-fact-crit}, we proved that if $\F$ is a
critical family, then there always exists an $\F$-positively-critical
 optimal solution (Corollary \ref{cor:combinatorial-positively-fitting-dual}). This
 argument did not use uniqueness. Indeed, it will also be used  to derive
 Lemma~\ref{lem:make-unique}, showing that a perturbation of the original integer cost
 function satisfies our uniqueness assumption (\ref{prop:uniqueness}).
We will need the following simple claim.
\begin{claim}\label{claim:even-cycle}
For a graph $G=(V,E)$, let $a,b:E\rightarrow \R_+$ be two vectors on
the edges with $a(\delta(v))=b(\delta(v))$ for every $v\in V$.
If $a$ and $b$ are not identical, then there exists an even length
closed walk $C$ such that for every odd edge $e\in C$, $a(e)>0$ and for every
even edge $e\in C$, $b(e)>0$.
\end{claim}
\begin{proof}
Due to the degree constraints, $z=a-b$ satisfies $z(\delta(v))=0$ for
every $v\in V$, and since $a$ and $b$ are not identical, $z$ has
nonzero components.
If there is an edge $uv\in E$ with $z(uv)>0$ then there must be
another edge $uw\in E$ with $z(uw)<0$. This implies the existence of
an alternating even closed walk $C$ where for every odd edge $e\in C$,
$0<z(e)=a(e)-b(e)$,
and for every even edge, $0>z(e)=a(e)-b(e)$. This proves the claim.
\end{proof}

\begin{proof}[Proof of Lemma~\ref{lem:make-unique}]
Let $\tilde c$ denote the perturbation of the integer cost
$c:E\rightarrow \Z$. Consider a graph $G=(V,E)$, perturbed cost $\tilde c$ and
critical family $\F$ where (\ref{prop:uniqueness}) does not hold. Choose
a counterexample with $|\F|$ minimal. Let $x$ and $y$ be two different optimal
solutions to
$P_{\F}(G,\tilde c)$. Since $\F$ is a critical family, by Corollary
\ref{cor:combinatorial-positively-fitting-dual}, there exists an $\pfc{\F}$ dual
optimal solution, say $\Pi$.

First, assume $\F=\emptyset$. Then $x$ and $y$ are both optimal
solutions to the bipartite relaxation $P_0(G,\tilde c)$. As they are
not identical, Claim~\ref{claim:even-cycle} gives an even closed walk  $C$ such that
$x(e)>0$ on every even edge and $y(e)>0$ on every odd edge. Let $\gamma_1$ and
$\gamma_2$ be the sum of edge costs on even and on odd edges of $C$, respectively. Then
for some $\varepsilon>0$, we could modify $x$ by decreasing $x(e)$ by $\varepsilon$ on
even edges and increasing on odd edges, and $y$ the other way around. These give two
other optimal matchings $\bar x$ and $\bar y$, with $\tilde c^T\bar x=\tilde
c^Tx+(\gamma_2-\gamma_1)\varepsilon$ and $\tilde c^T\bar y=\tilde
c^Ty+(\gamma_1-\gamma_2)\varepsilon$. Since $\bar x$ and $\bar y$ are both optimal,
this gives $\gamma_1=\gamma_2$. However, the fractional parts of $\gamma_1$ and
$\gamma_2$ must be different according to the definition of the perturbation, giving a
contradiction.

The case $\F\neq\emptyset$ is slightly more complicated. First we will replace $x$, $y$ by
another pair of non-identical optimal solutions $a$, $b$.
We will be able to identify an alternating closed walk $C$ in  $\supp(a)\cup
\supp(b)$ with the additional property that if $C$ intersects
$\delta(S)$ for some set $S\in \F$, then it does so in exactly one
even and one odd edge. The modifications $\bar a$ and $\bar b$ defined
as above would again be feasible, implying $\gamma_1=\gamma_2$.

We first claim that $\Pi(S)>0$ must hold for all $S\in \F$. Indeed, if $\Pi(S)=0$ for
some $S\in \F$, then $x$ and $y$ would be two different  optimal solutions to
$P_{\F\setminus\{S\}}(G,\tilde c)$, contradicting the minimal choice of $\F$.
Let $T\in \F\cup\{V\}$ be a smallest set with the property that
$x(\delta(u,V-T))=y(\delta(u,V-T))$ for every $u\in T$, but $x$ and $y$ are not
identical inside $T$. Note that $V$ trivially satisfies this property and hence such a
set exists. Let ${\cal S}$ denote the collection of maximal sets $S$ in
$\F$ such that $S\subsetneq T$ (${\cal S}$ could possibly be empty).

Inside each maximal set $S\in{\cal S}$, we modify $x$ and $y$ such that they are still
both optimal and different from each other after the modification. Since $\Pi(S)>0$, we
have $x(\delta(S))=1$, $y(\delta(S))=1$ and $S$ is $\Pi$-factor-critical. For $u\in S$,
let $M_u$ denote the $\Pi$-critical matching for $u$ inside $S$.
Let $\alpha_u:=x(\delta(u,V-S)), \alpha_u'=y(\delta(u,V-S))$ and
$w:=\sum_{u\in S}\alpha_u M_u, w':=\sum_{u\in S}\alpha_u' M_u$. 
Consider the following modified solutions:
\begin{align*}
a(e)&:=
\begin{cases}
& x(e) \mbox{ if $e\in \delta(S)\cup E[V\setminus S]$},\\
& w(e) \mbox{ if $e\in E[S]$},
\end{cases}
\end{align*}
\begin{align*}
b(e)&:=
\begin{cases}
& y(e) \mbox{ if $e\in \delta(S)\cup E[V\setminus S]$},\\
& w'(e) \mbox{ if $e\in E[S]$},
\end{cases}
\end{align*}
We claim that $a$ and $b$ are both optimal and non-identical. Optimality follows since both of them use only tight edges w.r.t. $\Pi$,
and $a(\delta(Z))=b(\delta(Z))=1$ still holds for every $Z\in \F$. Further, if
$x$ and $y$ are not identical inside $S$ and $\alpha_u=\alpha_u'$ for
every $u\in S$, then $S$ contradicts the minimality of $T$. 
This leaves two cases: either {\em(1)} $x$ and $y$ are identical inside
$S$, and thus for every $u\in S$,  $\alpha_u=\alpha_u'$ or {\em(2)} $\alpha_u\neq
\alpha_u'$ for some $u\in S$. In both cases it follows that $a\neq b$.

\begin{claim}
Consider two edges $u_1v_1\in \delta(S)\cap \supp(a)$, $u_2v_2\in\delta(S)\cap
\supp(b)$, $u_1,u_2\in S$. Then there exists an even alternating path $P_S$ inside $S$
between $u_1$ and $u_2$ such that $a(e)>0$ for every even edge and $b(e)>0$ for every
odd edge. Also, consider the extended path $P_S'=u_1v_1P_Su_2v_2$. If there exists a
set $Z\subsetneq S,Z\in \F$ such that $V(P_S')\cap Z\neq \emptyset$, then $P_S'$
intersects $\delta(Z)$ in exactly one even and one odd edge.
\end{claim}
\begin{proof}
By the modification, $\supp(a)\cap E[S]$ contains the $\Pi$-critical-matching $M_{u_1}$
and $\supp(b)\cap E[S]$ contains the $\Pi$-critical-matching $M_{u_2}$. Then the
symmetric difference of $M_{u_1}$ and $M_{u_2}$ contains an
$u_1-u_2$ alternating path satisfying the requirements (by
Lemma~\ref{lem:critical-matching}).
\end{proof}

We perform the above modifications inside every $S\in {\cal S}$, and
denote by $a$ and $b$ the result of all these modifications; these are hence two non-identical optimal solutions.
Let us now contract all sets in ${\cal S}$ w.r.t. $\Pi$; let
$G',\tilde c',\F'$ denote the resulting graph, costs and laminar
family respectively. By Lemma~\ref{lem:contract-factor-critical}(ii),
the images $a'$ and $b'$ are both optimal solutions in
$P'_{\F'}(G',\tilde c')$.

We claim that $a'$ and $b'$ are still not identical inside
$T'$. Indeed, assume $a'$ and $ b'$ are identical inside $T'$. Then we must have had
case {\em(1)} for every $S\in {\cal S}$.
 Note that we only modified $x$ and $y$ on the edge set
$\cup_{S\in {\cal S}}E[S]$, and hence the contracted images must be
identical: $x'=a'$ and $y'=b'$, and consequently, $x'=y'$. As
$x$ and $y$ are not identical inside $T$, they must differ inside at least one $S\in
{\cal S}$.
This contradicts the fact that case {\em(1)} applied for every
$S\in{\cal S}$.

 By the assumption on $T$, we also have
 $a'(\delta(u',V'-T'))=b'(\delta(u',V'-T'))$ for every $u'\in T'$.
Hence Claim~\ref{claim:even-cycle} is applicable, giving an
alternating closed walk $C'$ with $a'(e)>0$ on every even edge and $b'(e)>0$
on every odd edge.
Now, we can extend $C'$ to an even alternating closed walk $C$ in the original graph
using the paths $P_S$ as in the above claim. The resulting closed walk $C$ will have
the property that if there exists a set $Z\subsetneq S,Z\in \F$ such that $V(C)\cap
Z\neq \emptyset$, then $C$ intersects $\delta(Z)$ in exactly one even and one odd
edge.
\end{proof}

\section{Open Questions}
Our initial motivation was to bound the number of iterations of the cutting plane
method using the Padberg-Rao procedure. This question remains open and any analysis
would have to deal with non-half-integral solutions.

Within our algorithm,  Lemma \ref{lem:pof-uniqueness-implies-half-integrality} shows
that it is sufficient to use positively-critical dual optimal solutions to maintain
proper-half-integrality. Can we prove efficient convergence of our cutting plane
algorithm using positively-critical dual optimal solutions (without using extremal dual
solutions)? We believe that such a proof of convergence should depend on whether the
following adversarial variation of Edmonds' algorithm for perfect matching is
polynomial time. Suppose we run the Edmonds' perfect matching algorithm, but after
every few primal/dual iterations, the adversary replaces the current dual solution with a different
one, still satisfying complementary slackness with the (unchanged) primal solution, and having
dual objective value at least as much as the previous one.

Given the encouraging results of this paper, it would be interesting
to prove efficient convergence of the cutting plane method for other
combinatorial polytopes. 
For example, one could try optimizing over the intersection of two matroid polytopes by a lazy-constraint approach---add violated rank constraints if the current LP optimum is not integral. 
Another direction could be to try this approach for optimizing over
the subtour elimination polytope.

\medskip

\noindent {\bf Acknowledgment.} We are grateful to William Cook and Richard Karp for
their kind help and encouragement. 

\bibliographystyle{abbrv}
\bibliography{matching}
\end{document}